\documentclass[11pt]{article}

\usepackage{amsmath}
\usepackage{amsfonts}
\usepackage{mathtools}
\usepackage{amsthm}
\usepackage{kantlipsum}
\usepackage{graphicx}

\usepackage{xcolor}
\usepackage{caption}
\usepackage{subcaption}
\usepackage{hyperref}
\usepackage{soul}
\usepackage{multirow}
\usepackage[round]{natbib}

\usepackage{float}
\floatstyle{plaintop}
\restylefloat{table}

\graphicspath{{./img/}}

\newtheorem{proposition}{Proposition}
\newtheorem{lemma}{Lemma}

\newenvironment{keywords}
  {\par\smallskip\noindent\footnotesize\textbf{Keywords: }\ignorespaces}
{\par\smallskip}

\newcommand{\E}[1]{\mathbb{E}\left(#1\right)}
\newcommand{\Esquare}[1]{\mathbb{E}\left[#1\right]}
\newcommand{\Ep}[2]{\mathbb{E}^{#2}\left(#1\right)}
\newcommand{\Epsquare}[2]{\mathbb{E}^{#2}\left[#1\right]}
\newcommand{\Es}[2]{\mathbb{E}_{#2}\left(#1\right)}
\newcommand{\Essquare}[2]{\mathbb{E}_{#2}\left[#1\right]}
\newcommand{\Var}[1]{\operatorname{Var}\left(#1\right)}
\newcommand{\Vars}[2]{\operatorname{Var}_{#2}\left(#1\right)}
\newcommand{\Varssquare}[2]{\operatorname{Var}_{#2}\left[#1\right]}
\newcommand{\Varp}[2]{\operatorname{Var}^{#2}\left(#1\right)}
\newcommand{\Cov}[2]{\operatorname{Cov}\left(#1, #2\right)}

\newcommand{\lhs}[1]{\widetilde{#1}}
\newcommand{\single}[1]{\{#1\}}

\def\cD{{\cal D}}

\def\cI{{\cal I}}
\def\cL{{\cal L}}

\def\cO{{\cal O}}

\def\cS{{\cal S}}

\def\cX{{\cal X}}
\def\cY{{\cal Y}}
\def\cU{{\cal U}}
\def\cu{{\tiny{\cal U}}}

\def\bA{{\bf A}}

\def\bX{{\bf X}}
\def\bx{{\bf x}}

\def\bZ{{\bf Z}}

\allowdisplaybreaks

\graphicspath{ {./img/}}

\usepackage[
margin=0.9in,
footskip=0.25in,
headsep=0.1in
]{geometry}
\usepackage{setspace}
\usepackage{etoolbox}
\usepackage{titlesec}
\AtBeginEnvironment{thebibliography}{
    \footnotesize
    \setstretch{0.9}
}

\titleformat{\section}{\normalfont\large\bfseries}{\thesection}{1em}{}
\titleformat{\subsection}{\normalfont\normalsize\bfseries}{\thesubsection}{1em}{}
\titlespacing*{\section}{0pt}{1.5ex plus 0.5ex minus .2ex}{0.8ex plus .2ex}
\titlespacing*{\subsection}{0pt}{1.2ex plus 0.4ex minus .2ex}{0.5ex plus .1ex}
\titlespacing*{\subsubsection}{0pt}{1.2ex plus 0.4ex minus .2ex}{0.5ex plus .1ex}

\appto\appendix{
	\titleformat{\section}
	{\normalfont\normalsize\bfseries}{\thesection}{1em}{}

	\titleformat{\subsection}
	{\normalfont\small\bfseries}
	{\thesubsection}
	{1em}
	{}

	\titleformat{\subsubsection}
	{\normalfont\small\bfseries}
	{\thesubsubsection}
	{1em}
	{}
}

\title{Nested Simulation Methods for Sobol' Index Estimation: Bias Correction, Budget Allocation, and Latin Hypercube Sampling}

\author{
Jingtao Zhang\thanks{This work was completed before Jingtao Zhang joined Amazon.}  \quad
Xi Chen\thanks{E-mail: \protect\nolinkurl{xchen6@vt.edu}.}  \\
\normalsize{Grado Department of Industrial and Systems  Engineering, Virginia Tech, USA}
}

\date{}

\begin{document}

\maketitle

\begin{abstract}

	Estimating the variance of a conditional expectation is a recurring problem in stochastic simulation, with applications in global sensitivity analysis and Sobol' index estimation. This paper revisits Sobol' index estimation through the lens of nested simulation and develops a unified comparison of classical pick--freeze estimators and nested simulation estimators under a common computational budget. We show that several standard pick--freeze estimators can be interpreted as nested simulation estimators with fixed inner-level sample sizes, enabling direct performance comparisons and clarifying their bias--variance behavior. Building on this perspective, we analyze the standard nested simulation estimator for the Sobol' index numerator and propose two jackknife-based extensions: an unbiased jackknife estimator and a split jackknife estimator that uses an independent preliminary sample to estimate the mean. Under crude Monte Carlo (CMC), the split jackknife estimator attains the canonical mean squared error (MSE) rate, whereas the standard nested simulation and unbiased jackknife estimators attain the slower nested simulation rate. We also characterize the associated allocations of outer- and inner-level simulation effort. Finally, we study the impact of Latin hypercube sampling (LHS), showing that it can improve the standard nested simulation estimator while undermining bias reduction in jackknife-based estimators unless the inner-level sample size grows with the total budget. Numerical experiments corroborate the theory and provide practical guidance on estimator selection for Sobol' index estimation under CMC and LHS.

\end{abstract}


\section{Introduction}

Estimating the variance of a conditional expectation is a fundamental problem in stochastic simulation, with applications in risk measurement, uncertainty quantification, and global sensitivity analysis. In this paper, we focus on Sobol' index estimation, where the numerator is the variance of the conditional expectation of a model output given a subset of input variables. Sobol' indices are widely used in global sensitivity analysis (GSA) to quantify how input uncertainty contributes to output variability \citep{sobol1990sensitivity}, with applications in epidemiological modeling \citep{kouye2022exploiting}, manufacturing \citep{liu2020identifying}, and sustainability analysis \citep{jaxa2021variance}. A large body of work has developed Monte Carlo (MC) estimators for Sobol' indices \citep{tarantola2007estimating,saltelli2010variance}. Many commonly used methods rely on the pick--freeze scheme \citep{saltelli2010variance,owen2013better,janon2014asymptotic}, which is straightforward to implement and attains the canonical mean squared error (MSE) rate $\cO(T^{-1})$ under a computational budget of $T$ model evaluations. Two prominent estimators in this class are due to \cite{janon2014asymptotic} and \cite{owen2013better}.

A complementary perspective is offered by \emph{nested simulation}, a general framework for estimating functionals of conditional expectations, including risk measures and variance-type quantities \citep{gordy2010nested,sun2011efficient}. Nested simulation generates an \emph{outer-level} sample of scenarios and, conditional on each outer-level scenario, draws \emph{inner-level} observations to approximate the corresponding conditional expectation. Under a fixed budget $T$, \cite{gordy2010nested} showed that the standard nested simulation estimator achieves its optimal MSE convergence rate of $\cO(T^{-2/3})$ using $\cO(T^{2/3})$ outer-level scenarios and $\cO(T^{1/3})$ inner-level observations per scenario. Although this rate is slower than the $\cO(T^{-1})$ rate attained by standard MC estimators for the mean of a random variable, bias-reduction techniques can improve nested simulation efficiency \citep{gordy2010nested,giles2019multilevel,liang2024fast}. Separately, within an analysis of variance (ANOVA) framework, \cite{sun2011efficient} derived an unbiased ``$1\frac{1}{2}$-level'' estimator for the variance of a conditional expectation that can achieve an $\cO(T^{-1})$ MSE convergence rate.

The connection between Sobol' index estimation and nested simulation has recently become more explicit. For example, \cite{goda2017computing} studied the pick--freeze estimators of \cite{janon2014asymptotic} and introduced a class of non-nested estimators that clarify the relationship between GSA and the nested simulation problem of estimating the variance of a conditional expectation. This perspective suggests that tools from nested simulation, including estimator construction, bias reduction, and budget allocation, can inform Sobol' index estimation, while classical pick--freeze estimators can also be viewed through a nested simulation lens. However, systematic comparisons of these estimators under a common computational budget remain limited.

Another practical ingredient in Sobol' index estimation is \emph{Latin hypercube sampling} (LHS), a widely used variance-reduction technique in MC integration, experimental design, and uncertainty quantification \citep{mckay1979comparison}. In GSA, LHS is often used to improve the empirical performance of Sobol' index estimators relative to crude Monte Carlo  \citep[CMC,][]{janon2014asymptotic,kouye2022exploiting,puy2022comprehensive}. Recent work has proposed more efficient LHS-based designs to reduce sampling effort in numerical studies \citep{gilquin2019making,ehre2020framework,damblin2021adaptive}. From a theoretical perspective, \cite{tissot2012estimating} and \cite{gilquin2016recursive} analyzed pick--freeze Sobol' index estimators under LHS, showing that these estimators can achieve lower variance than under CMC.
However, the effect of LHS on nested simulation estimators is less straightforward than in standard MC settings. This motivates a systematic analysis of how LHS interacts with estimator structure in nested simulation.

In this paper, we develop a unified analysis of nested simulation methods for Sobol' index estimation under a fixed computational budget, considering both CMC and LHS.
 We first show that classical pick--freeze estimators, including those of \cite{janon2014asymptotic} and \cite{owen2013better}, can be interpreted as nested simulation estimators with fixed inner-level sample sizes. This perspective places pick--freeze, standard nested simulation, and bias-corrected nested simulation estimators within a common framework and enables direct comparison of their bias, variance, and MSE behavior under the same computational budget. For the standard nested simulation estimator, we characterize the bias--variance trade-off and derive the asymptotically optimal allocation of outer- and inner-level simulation effort.

Building on this framework, we propose two jackknife-based nested simulation estimators. The first is an unbiased jackknife estimator that removes the leading bias of the standard nested simulation estimator. The second is a split jackknife estimator that uses an independent preliminary sample to estimate the mean. Under CMC, the split jackknife estimator attains the $\cO(T^{-1})$ MSE rate, whereas the standard nested simulation and unbiased jackknife estimators attain the slower nested simulation rate. This comparison highlights the role of mean estimation: unlike settings with a known centering constant, the Sobol' index numerator requires estimating the mean of the conditional expectation, and this additional step affects the dependence structure and variance behavior of bias-corrected estimators.

We further characterize how LHS affects these estimator families. While LHS does not improve the convergence rates of pick--freeze estimators, it can improve the bias and variance behavior of the standard nested simulation estimator, particularly for first-order Sobol' indices. By contrast, LHS can undermine the bias cancellation of jackknife-based nested simulation estimators and the ``$1\frac{1}{2}$-level'' estimator unless the inner-level sample size grows with the total budget. Numerical experiments corroborate these findings and yield practical recommendations for choosing among pick--freeze, standard nested simulation, and bias-reduced nested simulation estimators under CMC and LHS.

The remainder of the paper is organized as follows. Section \ref{sec:sobol_index} reviews Sobol' indices and their estimation via the pick--freeze scheme. Section \ref{sec:nest} studies nested simulation estimators, including the proposed jackknife estimators, under CMC. Section \ref{sec:ns_lhs} investigates the impact of LHS. Section \ref{sec:experiment} reports numerical experiments, and Section \ref{sec:conclusion} concludes the paper.

\section{Review of Sobol' Indices and the Pick-Freeze Scheme}\label{sec:sobol_index}
This section provides a brief overview of Sobol' indices for global sensitivity analysis and reviews two representative estimators based on the pick-freeze scheme.

\subsection{Sobol' Indices for Global Sensitivity Analysis}\label{subsec:Sobol_index_review}
GSA quantifies how uncertainty in model inputs propagates to variability in the output. Among variance-based GSA methods, Sobol' indices are widely used. Let $\cX \subset \mathbb{R}^{p}$ denote the $p$-dimensional input space, and consider a computational model $\cY=f(\bX)$, where $f:\cX\mapsto\mathbb{R}$ maps the input vector $\bX=(X_1,X_2,\dots,X_p)^{\top}\in\cX$ to the scalar output $\cY$. Throughout the paper, we assume that the fourth moment of the output is bounded, i.e., $\E{\cY^{4}}<\infty$. Let $\cu\subset [p]$, where $[p]\coloneqq\{1,2,\dots,p\}$, denote an index set of inputs. Define $\bX_{\cu}$ as the subvector of $\bX$ indexed by $\cu$ (e.g., if $\cu=\{1,2\}$ then $\bX_{\cu}=(X_1,X_2)^{\top}$), and let $\bX_{-\cu}\coloneqq \bX\setminus \bX_{\cu}$.

Sobol' indices quantify the contribution of input variables to the output variance through the functional ANOVA decomposition,
\[
\Var{\cY}=\sum_{i=1}^{p}\sum_{|\cu|=i}V_{\cu}(\cY) \ ,
\]
where $V_{i}(\cY) \coloneqq \Var{\E{\cY \mid X_i}}, V_{\{i,j\}}(\cY) \coloneqq \Var{\E{\cY \mid X_i, X_j}} - V_{i}(\cY)-  V_{j}(\cY)$ for $i \neq j$, and so forth.  Formally, the Sobol' index associated with $\bX_{\cu}$ is
\begin{equation}\label{eq:sobol}
    S^{\cu} \coloneqq \frac{\Var{\E{\cY \mid \bX_{\cu}}}}{\Var{\cY}} \ .
\end{equation}
In particular, when $\cu$ consists of a single input variable (i.e., $|\cu|=1$), $S^{\cu}$ is referred to as the first-order Sobol' index. Sobol' indices take values in $[0,1]$, with larger values indicating a stronger influence of the corresponding input(s) on the model output. While the denominator $\Var{\cY}$ in \eqref{eq:sobol} can be readily estimated via MC simulation, estimating the numerator $\Var{\E{\cY \mid \bX_{\cu}}}$ is more challenging. The remainder of this work focuses on estimating this quantity; for brevity, we denote it by $V$ in what follows.

\subsection{Pick-Freeze Scheme and Estimators}\label{subsec:pick_freeze}
The pick-freeze scheme underlies many variance-based GSA methods for estimating Sobol' indices. The key idea is to \emph{freeze} one or more input variables and \emph{pick} (i.e., randomly sample) the remaining variables to assess the induced variability in the output. This enables one to isolate and quantify the contribution of the frozen input variable(s) to the output variance and thereby estimate $V$.

Specifically, the pick-freeze scheme exploits the identity (see Lemma 2.2 of \cite{janon2014asymptotic})
\[
V = \Cov{\cY}{\cY^{\prime}} = \E{\cY \cY^{\prime}} - \E{\cY}\E{\cY^{\prime}} \ ,
\]
where $\cY=f(\bX_{\cu},\bX_{-\cu})$, $\cY'=f(\bX_{\cu},\bX_{-\cu}')$, and $\bX_{-\cu}'$ is an independent copy of $\bX_{-\cu}$.  In practice, one draws independent realizations of $\bX_{-\cu}$ conditional on $\bX_{\cu}$, evaluates the model at input vectors $(\bX_{\cu},\bX_{-\cu})$ and $(\bX_{\cu},\bX_{-\cu}')$, and uses the resulting paired outputs to estimate $V$.

A classical pick-freeze estimator studied in \cite{ishigami1990importance}, \cite{janon2014asymptotic} , and \cite{saltelli2010variance} is widely used due to its simplicity and practical effectiveness. Let $K$ denote the sample size, and consider two sets of $p$-dimensional input vectors,
$\{(\bX_{\cu,i},\bX_{-\cu,i})\!:i\in [K]\}$ and $\{(\bX_{\cu,i},\bX_{-\cu,i}')\!:i\in [K]\}$,
where $\bX_{\cu,i}$ (respectively, $\bX_{-\cu,i}$) is the $i$th realization of $\bX_{\cu}$ (resp., $\bX_{-\cu}$), and $\bX_{-\cu,i}'$ is an independent realization of $\bX_{-\cu}$ conditional on $\bX_{\cu}$.
The classical pick-freeze (\textbf{PF}) estimator of $V$ is
\begin{equation}\label{eq:pick_freeze}
	V_{\mbox{\tiny PF}}^{\cu}
	=\frac{1}{K}\sum_{i=1}^{K} f(\bX_{\cu,i},\bX_{-\cu,i})\, f(\bX_{\cu,i},\bX_{-\cu,i}')
	-\left(\frac{1}{2K}\sum_{i=1}^{K}\Big[f(\bX_{\cu,i},\bX_{-\cu,i})+f(\bX_{\cu,i},\bX_{-\cu,i}')\Big]\right)^2 \  ,
\end{equation}
which is widely adopted owing to its ease of implementation and its effectiveness across a broad
range of applications.

A notable limitation of $V_{\mbox{\tiny PF}}^{\cu}$ in \eqref{eq:pick_freeze} is its poor performance when $V$ is small (and hence $S^{\cu}$ is small) \citep{owen2013better}. This stems from the fact that the two terms in $V=\E{\cY \cY^{\prime}} - \E{\cY}\E{\cY^{\prime}}$ can be close in magnitude; subtracting them may lead to catastrophic cancellation, degrading numerical stability and estimator reliability.
To improve estimation when Sobol' indices are small, \cite{owen2013better} proposed the ``Correlation 2'' (\textbf{CR}) method for estimating $V$, which is particularly effective when $S^{\cu}$ is small (e.g., below $0.1$).  Although the CR method retains the pick-freeze structure, it uses four sets of input vectors rather than two. Specifically, the CR estimator is
\begin{equation}\label{eq:corr2}
	V_{\mbox{\tiny CR}}^{\cu}
	= \frac{1}{K} \sum_{i=1}^K
	\left( f(\bX_{\cu,i}, \bX_{-\cu,i}) - f(\bX_{\cu,i}^{\prime\prime}, \bX_{-\cu,i}) \right) \cdot
	\left( f(\bX_{\cu,i}, \bX_{-\cu,i}^{\prime}) - f(\bX_{\cu,i}^{\prime}, \bX_{-\cu,i}^{\prime}) \right)\ ,
\end{equation}
where $\bX_{\cu,i}^{\prime}$ and $\bX_{\cu,i}^{\prime\prime}$ are independent realizations of $\bX_{\cu}$ for $i\in [K]$, and the remaining subvectors are defined as in \eqref{eq:pick_freeze}.

Several remarks on the properties of these two pick-freeze estimators are in order. \cite{goda2017computing} showed that the PF estimator in \eqref{eq:pick_freeze} is biased, with bias decaying at rate $\cO(K^{-1})$ (see Theorem~1 in \cite{goda2017computing}). In contrast,  \cite{owen2013better} established that the CR estimator is unbiased (see Theorem~5.1 in \cite{owen2013better}). Moreover, the structures of both estimators imply variances of order $\cO(K^{-1})$, and thus mean squared errors (MSEs) of order $\cO(K^{-1})$.
It is important to note that these properties are derived under CMC sampling. Empirical studies have further indicated that LHS can improve efficiency in practice \citep{janon2014asymptotic, kouye2022exploiting, puy2022comprehensive}. From a theoretical standpoint, although no faster MSE rate has been established to the best of our knowledge, LHS has been shown to achieve a smaller asymptotic variance \citep{gilquin2016recursive,tissot2012estimating}.

\section{Nested Simulation Estimation of the Variance of a Conditional Expectation Under Crude Monte Carlo}\label{sec:nest}

This section develops a nested simulation framework under CMC. Subsection~\ref{subsec:standard_ns} presents new bias--variance and MSE analyses for the standard nested simulation estimator and relates this perspective to existing pick-freeze estimators. Subsection~\ref{subsec:jackknife} then introduces two new jackknife estimators built on the nested simulation framework and establishes their theoretical properties.

\subsection{Nested Simulation Framework and Connections to Existing Estimators}\label{subsec:standard_ns}

Nested simulation has been widely studied for estimating functionals of conditional expectations. In this framework, the object of interest can often be written as $\rho(L(\bX_\cu))$, where $L(\bX_\cu):=\E{\cY\mid \bX_\cu}$ and $\rho$ is a real-valued functional defined on a suitable space of random variables. For example, \cite{gordy2010nested}, \cite{zhang2022bootstrap}, and \cite{liang2024fast} investigated nested simulation estimators for functionals of the form $\E{g(L(\bX_\cu))}$ with a given function $g$, often with a known centering constant or threshold. In contrast, the numerator of the Sobol' index considered here is
	$
	V=\Var{L(\bX_\cu)}
	=\E{\left(L(\bX_\cu)-\E{L(\bX_\cu)}\right)^2}$,
	where the centering term $\E{L(\bX_\cu)}=\E{\cY}$ is unknown and must be estimated from simulation outputs. This additional estimation step is central to the bias--variance behavior of the nested simulation estimators studied below.

Given a total computational budget $T=KN$, where $K$ and $N$ denote the outer- and inner-level sample sizes, respectively, the standard nested simulation under CMC proceeds as follows. First, we generate $K$ independent and identically distributed (i.i.d.)\ outer-level scenarios $\{\bX_{\cu,i}\}_{i=1}^K$. Conditional on each $\bX_{\cu,i}$, we then generate $N$ i.i.d.\ inner-level input subvectors $\{\bX_{-\cu,j}^{(i)}\}_{j=1}^N$  and compute the corresponding outputs $\cY_{ij}\coloneqq f(\bX_{\cu,i},\bX_{-\cu,j}^{(i)})$, where $i\in[K]$ and $j\in[N]$.
Define the inner-level sample average output
$
L_N(\bX_{\cu,i})\coloneqq N^{-1}\sum_{j=1}^N \cY_{ij},
$
which serves as an estimator of $L(\bX_{\cu,i})\coloneqq \E{\cY_{ij}\mid \bX_{\cu,i}}$ for each $i \in [K]$.
The nested simulation (\textbf{NS}) estimator of $V$ is the sample variance of $\{L_N(\bX_{\cu,i})\}_{i=1}^K$:
\begin{equation}\label{eq:nest}
	V_{\mbox{\tiny NS}}^{\cu}
	=\frac{1}{K-1}\sum_{i=1}^{K}\left(L_N(\bX_{\cu,i})-\frac{1}{K}\sum_{i=1}^{K}L_N(\bX_{\cu,i})\right)^2 .
\end{equation}
To reduce MSE effectively, one typically requires $K,N\to\infty$ as $T\to\infty$, with a budget allocation that balances bias and variance.
We analyze $\mathrm{MSE}(V_{\mbox{\tiny NS}}^{\cu})$ via
\begin{equation}\label{eq:mse_decompose}
\mbox{MSE}(V_{\mbox{\tiny NS}}^{\cu}) = \Esquare{(V_{\mbox{\tiny NS}}^{\cu}-V)^2} = \bigg(\E{V_{\mbox{\tiny NS}}^{\cu}} - V\bigg)^2 + \Var{V_{\mbox{\tiny NS}}^{\cu}} \ .
\end{equation}
The following result summarizes the asymptotic properties of the bias and variance; the proof is deferred to Appendix \ref{app:bias_variance_proof}.

\begin{proposition}\label{prop:bias_variance}
	The bias and variance of the NS estimator satisfy
	\[
	\E{V_{\mbox{\tiny NS}}^{\cu}}-V
	=N^{-1}\big(\Var{\cY}-V\big),
	\qquad
	\Var{V_{\mbox{\tiny NS}}^{\cu}}
	=K^{-1}\,(\kappa_{\E{\cY \mid \bX_\cu}}-1)\,V^2+o\!\left(K^{-1}\right) \ ,
	\]
	where $\kappa_{\E{\cY \mid \bX_\cu}}$ denotes the kurtosis of $\E{\cY\mid \bX_\cu}$.
\end{proposition}

From Proposition \ref{prop:bias_variance}, we obtain the normalized MSE as
	\begin{equation}\label{eq:mse}
	\frac{\mbox{MSE}(V_{\mbox{\tiny NS}}^{\cu})}{\Varp{\cY}{2}} = K^{-1} \cdot (\kappa_{\E{\cY \mid \bX_\cu}} - 1) \cdot \left( S^{\cu} \right)^2   + N^{-2} \cdot \left(1 - S^{\cu} \right)^2 +o\!\left(K^{-1}\right) \ .
    \end{equation}
Consequently, the asymptotically optimal allocation that minimizes the normalized MSE is
\begin{equation}\label{eq:opt_allocation}
	K^{\ast} = \left(\frac{(\kappa_{\E{\cY \mid \bX_\cu}}-1)\left( S^{\cu} \right)^2}{2(1-S^{\cu})^2}\right)^{1/3} \cdot  T^{2/3}, \quad N^{\ast} = \left(\frac{2(1-S^{\cu})^2}{(\kappa_{\E{\cY \mid \bX_\cu}}-1)\left( S^{\cu} \right)^2} \right)^{1/3} \cdot T^{1/3} \ ,
\end{equation}
where $K^{\ast}$ and $N^{\ast}$ are the asymptotically optimal outer- and inner-level sample sizes for constructing $V_{\mbox{\tiny NS}}^{\cu}$.
Equation \eqref{eq:opt_allocation} shows that the optimal allocation depends on the Sobol' index $S^{\cu}$. When $S^{\cu}$ is close to $0$, $\bX_\cu$ contributes little to $\Var{\cY}$, and the optimal $K^{\ast}$ is relatively small. When $S^{\cu}$ is close to $1$, a larger $K^{\ast}$ is preferred to resolve variation across outer-level scenarios. The dependence of the inner-level sample size $N^{\ast}$ on $S^{\cu}$ follows the same reasoning.

To implement \eqref{eq:opt_allocation}, one must estimate $\kappa_{\E{\cY \mid \bX_\cu}}$ and $S^{\cu}$. We do so via a pilot experiment. Specifically, given a total budget $T$, allocate $2\alpha T$ to the pilot stage, where $\alpha\in(0,0.5)$ and $\alpha T$ is assumed to be an integer for simplicity. Using $\{(\cY_i,\cY_i') : i\in[\alpha T]\}$ with $\cY_i=f(\bX_{\cu,i},\bX_{-\cu,i})$ and $\cY_i'=f(\bX_{\cu,i},\bX_{-\cu,i}')$, define $\overline{\cY}_i\coloneqq(\cY_i+\cY_i')/2$ and $\overline{\overline{\cY}}\coloneqq(\alpha T)^{-1}\sum_{i=1}^{\alpha T}\overline{\cY}_i$. We estimate $\Var{\cY}$ and $\kappa_{\E{\cY \mid \bX_\cu}}$ by
\[
\widehat{\operatorname{Var}}(\cY) = \frac{1}{\alpha T}\sum_{i=1}^{\alpha T} \frac{\cY_i^2+ (\cY^{\prime}_i)^2}{2} - \left(\overline{\overline{\cY}}\right)^2  \mbox{ and } \  \widehat{\kappa}_{\E{\cY \mid \bX_\cu}}  =  \frac{\alpha T \cdot\sum_{i=1}^{\alpha T}\left(\overline{\cY}_i-\overline{\overline{\cY}}\right)^4}{ \left(\sum_{i=1}^{\alpha T}\left(\overline{\cY}_i-\overline{\overline{\cY}}\right)^2\right)^{2}} \ ,
\]
and then estimate $S^{\cu}$ via $\widehat{S}^{\cu}=V_{\mbox{\tiny PF}}^{\cu}/\widehat{\operatorname{Var}}(\cY)$, where $V$ is estimated using the PF estimator in \eqref{eq:pick_freeze}.
The optimal outer- and inner-level sample sizes $K^{\ast}$ and $N^{\ast}$ are then obtained by plugging $\widehat{S}^{\cu}$ and $\widehat{\kappa}_{\E{\cY \mid \bX_\cu}}$ into \eqref{eq:opt_allocation}, and the remaining budget $(1-2\alpha)T$ is allocated to construct $V_{\mbox{\tiny NS}}^{\cu}$ in \eqref{eq:nest}. Numerical experiments indicate that a small pilot proportion (e.g., $\alpha=0.05$) is sufficient to estimate both $\kappa_{\E{\cY \mid \bX_\cu}}$ and $S^{\cu}$ and to provide reliable estimates of $K^{\ast}$ and $N^{\ast}$ for constructing the NS estimator.

Equations \eqref{eq:mse}--\eqref{eq:opt_allocation} imply that $V_{\mbox{\tiny NS}}^{\cu}$ attains an MSE rate of $\cO(T^{-2/3})$ under CMC. Under the same CMC setting, \cite{sun2011efficient} proposed a ``$1\frac{1}{2}$-level'' nested simulation estimator within the ANOVA framework for estimating $V$, which we refer to as the \textbf{OH} estimator. To facilitate further discussion, we express it in terms of $V_{\mbox{\tiny NS}}^{\cu}$ as
	\begin{equation}\label{eq:one-half}
		V_{\mbox{\tiny OH}}^{\cu} = V_{\mbox{\tiny NS}}^{\cu} - \frac{1}{K(N-1)}\sum_{i=1}^{K}\sum_{j=1}^{N}(\cY_{ij} - L_N(\bX_{\cu, i}))^2 \ .
	\end{equation}
	\cite{sun2011efficient} showed that $V_{\mbox{\tiny OH}}^{\cu}$ is unbiased and satisfies $\Var{V_{\mbox{\tiny OH}}^{\cu}}=\cO(K^{-1})$. Furthermore, $V_{\mbox{\tiny OH}}^{\cu}$ achieves its optimal MSE rate of $\cO(T^{-1})$ under the allocation $K^{\ast}=T/N^{\ast}$ with a bounded inner-level sample size $N^{\ast}<\infty$, improving upon the nested simulation estimator $V_{\mbox{\tiny NS}}^{\cu}$ in \eqref{eq:nest}. This improvement stems from the second term on the right-hand side (RHS) of \eqref{eq:one-half}, which corrects the bias of $V_{\mbox{\tiny NS}}^{\cu}$ without altering the $\cO(K^{-1})$ variance decay. \cite{sun2011efficient} also described how to select the optimal inner-level sample size given a total budget $T$ via a pilot experiment.

Viewed through a nested simulation lens, the pick-freeze estimators in Subsection~\ref{subsec:pick_freeze} correspond to using a fixed inner-level sample size (two for PF and four for CR). Under a total budget $T$, this yields outer-level sample sizes $K\approx T/2$ (PF) and $K\approx T/4$ (CR), ignoring integrality. In light of the discussion at the end of Section~\ref{sec:sobol_index}, both estimators achieve an MSE convergence rate of $\cO(T^{-1})$ under CMC.

We close this subsection by noting that the slower $\cO(T^{-2/3})$ MSE convergence rate of $V_{\mbox{\tiny NS}}^{\cu}$ in \eqref{eq:nest} is primarily due to its $\cO(N^{-1})$ bias, which forces the inner-level sample size $N$ to grow with the budget $T$ and diverts budget that could otherwise increase the outer-level sample size $K$. This observation motivates bias-reduction techniques that can potentially improve the efficiency of the NS estimator.

\subsection{New Jackknife Estimators}\label{subsec:jackknife}

Following the discussion in Subsection~\ref{subsec:standard_ns}, we propose two new jackknife estimators for the variance of a conditional expectation under the nested simulation framework with CMC sampling.

\subsubsection{Unbiased Jackknife Estimator}\label{subsec:unbias_jack}

Under the nested simulation framework, we construct a jackknife estimator based on the NS estimator in \eqref{eq:nest}. For each outer-level scenario $\bX_{\cu,i}$ with $i \in [K]$, the output sample $\{\cY_{ij}\}_{j \in [N]}$ associated with the inner-level sampling is partitioned into $I$ nonoverlapping sections, where $I$ is a fixed integer satisfying $2 \leq I \leq N$, and $N$ is divisible by $I$ \citep{gordy2010nested, liang2024fast}.
Define
$\cI_l \coloneqq \{j \in [N]: (l-1)\cdot (N/I) + 1 \leq j \leq l\cdot (N/I)\}$
as the index set of the $l$th section for each $l \in [I]$. Let $V^{\cu}_{\mbox{\tiny NS}, -l}$ denote the NS estimator constructed according to \eqref{eq:nest} after omitting the outputs in the $l$th section, $\{\cY_{ij}\}_{j \in \cI_l}$, for each outer-level scenario $\bX_{\cu,i}$. We then define the jackknife (\textbf{JK}) estimator for $V$ as
\begin{equation}\label{eq:jack}
	V^{\cu}_{\mbox{\tiny JK}} = I \cdot V_{\mbox{\tiny NS}}^{\cu} -    \frac{(I-1)}{I}\sum_{l=1}^{I}V^{\cu}_{\mbox{\tiny NS},-l} \ .
\end{equation}
The following result summarizes the properties of $V_{\mbox{\tiny JK}}^{\cu}$, with the proof provided in Appendix \ref{app:jack_proof}.
\begin{proposition}\label{prop:jack}
The jackknife estimator is unbiased, i.e., $\E{V_{\mbox{\tiny JK}}^{\cu}} - V= 0$, and its variance satisfies $\Var{V_{\mbox{\tiny JK}}^{\cu}} = aK^{-1} + bN^{-2} + o(K^{-1}) + o(N^{-2})$, where $a$ and $b$ are positive constants.
\end{proposition}

Proposition \ref{prop:jack} implies that, although $V_{\mbox{\tiny JK}}^{\cu}$ is unbiased for $V$, its variance is affected by the bias-reduction terms $V_{\mbox{\tiny NS},-l}^{\cu}$ for $l \in [I]$ in \eqref{eq:jack}. This variance behavior leads to an asymptotically optimal budget allocation $K^{\ast}=\cO(T^{2/3})$ and $N^{\ast}=\cO(T^{1/3})$, and hence the MSE convergence rate of $V_{\mbox{\tiny JK}}^{\cu}$ remains $\cO(T^{-2/3})$, the same as $V^{\cu}_{\mbox{\tiny NS}}$. In other words, the jackknife construction removes bias but does not improve the MSE convergence rate. This may appear to conflict with \cite{liang2024fast}, who reported an MSE convergence rate of at least $\cO(T^{-4/5})$ for their jackknife estimator. The difference stems from the estimation target: Liang et al.\ (2024) considered $\E{(\E{\cY \mid \bX_\cu} - z)^2}$ with a known constant $z$, whereas here $z=\E{\cY}$ is unknown and must be estimated.
In our setting, $\E{\cY}$ is estimated in $V_{\mbox{\tiny NS}}^{\cu}$ (resp., in $V^{\cu}_{\mbox{\tiny NS},-l}$) by $K^{-1}\sum_{i=1}^{K}L_N(\bX_{\cu, i})$ (resp., $K^{-1}\sum_{i=1}^{K}L_{N, -l}(\bX_{\cu, i})$). Here, $L_N(\bX_{\cu, i})$ and $L_{N, -l}(\bX_{\cu, i})$ estimate $\E{\cY \mid \bX_{\cu,i}}$ for each $i \in [K]$, and $L_{N,-l}(\bX_{\cu,i})$ has the same functional form as $L_{N}(\bX_{\cu,i})$ but excludes the outputs in the $l$th section.
Estimating $\E{\cY}$ and $\E{\cY \mid \bX_{\cu,i}}$ from overlapping observations induces additional correlations between $V^{\cu}_{\mbox{\tiny NS}}$ and $V^{\cu}_{\mbox{\tiny NS},-l}$ for each $l \in [I]$ in \eqref{eq:jack}, compared with the estimation target considered in \cite{liang2024fast}. As a result, the best attainable variance decay for $V^{\cu}_{\mbox{\tiny JK}}$ matches that of $V^{\cu}_{\mbox{\tiny NS}}$, achieved under the optimal allocation $K^{\ast}=\cO(T^{2/3})$ and $N^{\ast}=\cO(T^{1/3})$.

Implementing the asymptotically optimal allocation for the JK estimator in \eqref{eq:jack} requires estimating the constants $a$ and $b$. Related problems have been studied via bootstrap-based methods \citep{zhang2022bootstrap,liang2024fast}, but in our setting the estimation of $a$ and $b$ is more involved and cannot be handled directly by their approaches. We leave this issue for future research. In the present study, we adopt the scaling $K^{\ast}=\cO(T^{2/3})$ and $N^{\ast}=\cO(T^{1/3})$ for $V_{\mbox{\tiny JK}}^{\cu}$ in our numerical experiments.

\subsubsection{Split Jackknife Estimator}\label{subsec:split_jack}

Motivated by the analysis in Subsection~\ref{subsec:unbias_jack}, we propose the split jackknife (\textbf{SJ}) estimator. The SJ estimator follows the same jackknife construction as $V_{\mbox{\tiny JK}}^{\cu}$ in \eqref{eq:jack}, with one key modification: instead of using $K^{-1}\sum_{i=1}^{K}L_N(\bX_{\cu, i})$ and $K^{-1}\sum_{i=1}^{K}L_{N, -l}(\bX_{\cu, i})$ (for $l \in [I]$) to estimate $\E{\cY}$, we estimate $\E{\cY}$ using an independent dataset. This modification helps reduce the correlation among the components of the JK estimator.

Specifically, the total budget is split to generate two datasets: a preliminary dataset $\cD_{\mbox{\tiny pre}}$ for estimating the mean $\E{\cY}$ and an estimation dataset $\cD_{\mbox{\tiny est}}$ for estimating $V$.
Let $\widehat{\mu} \coloneqq J^{-1}\sum_{j=1}^{J}\cY_j$ denote the sample mean computed from $\cD_{\mbox{\tiny pre}}$, where $J \coloneqq |\cD_{\mbox{\tiny pre}}|$. Following a structure analogous to the JK estimator, the SJ estimator is defined as
\begin{align}\label{eq:fse}
	V^{\cu}_{\mbox{\tiny SJ}} &= \frac{I}{K}\sum_{i=1}^{K} (L_N(\bX_{\cu,i}) - \widehat{\mu})^2 - \frac{I-1}{I} \sum_{l=1}^{I} \frac{1}{K}\sum_{i=1}^{K}(L_{N,-l}(\bX_{\cu,i})-\widehat{\mu})^2  \nonumber \\
	&= \frac{1}{K}\sum_{i=1}^{K}  \left( I \cdot (L_N(\bX_{\cu,i}) - \widehat{\mu})^2 - \frac{I-1}{I} \sum_{l=1}^{I} (L_{N,-l}(\bX_{\cu,i})-\widehat{\mu})^2  \right) \nonumber \\
	&\coloneqq \frac{1}{K}\sum_{i=1}^{K}V_{\mbox{\tiny SJ},i}^{\cu} \ .
\end{align}

A key advantage of the SJ estimator is that, conditional on $\cD_{\mbox{\tiny pre}}$, the per-scenario contributions $V_{\mbox{\tiny SJ},i}^{\cu}$ in \eqref{eq:fse} are independent across $i$. This independence enables $V^{\cu}_{\mbox{\tiny SJ}}$ to achieve a faster variance decay than the JK estimator $V^{\cu}_{\mbox{\tiny JK}}$ in \eqref{eq:jack} by choosing the size of $\cD_{\mbox{\tiny pre}}$ appropriately. The following result summarizes the properties of $V^{\cu}_{\mbox{\tiny SJ}}$, with the proof provided in Appendix \ref{app:new_jack_proof}.

\begin{proposition}\label{prop:new_jack}
	The SJ estimator $V^{\cu}_{\mbox{\tiny SJ}}$ satisfies $\E{V^{\cu}_{\mbox{\tiny SJ}}} - V= J^{-1} \Var{\cY}$ and $\Var{V^{\cu}_{\mbox{\tiny SJ}}} = cK^{-1} + dJ^{-2}+ o(K^{-1})  + o(J^{-2})$, where $c$ and $d$ are positive constants.
\end{proposition}

Proposition \ref{prop:new_jack} shows that the bias and variance of the SJ estimator $V_{\mbox{\tiny SJ}}^{\cu}$ are independent of the inner-level sample size $N$. Instead, the variance decreases with the outer-level sample size $K$ and the preliminary dataset size $J$, while the bias decreases only with $J$.
This decoupling is advantageous under a total budget $T = J + KN$: increasing $J$ need not substantially restrict the choice of $K$ and $N$. By fixing $N$ and choosing $J$ to grow at the same rate as $K$, $V_{\mbox{\tiny SJ}}^{\cu}$ attains an $\cO(T^{-1})$ MSE convergence rate, improving upon both the NS estimator in \eqref{eq:nest} and the JK estimator in \eqref{eq:jack}. More broadly, this result suggests that, within nested simulation, decoupling the estimation of $\E{\cY}$ from that of $\E{\cY \mid \bX_{\cu}}$ can be an effective way to mitigate dependence and improve efficiency.

We close this section by noting that all theoretical results above are derived under CMC sampling at both the outer and inner levels of the nested simulation framework. While CMC serves as a natural baseline, alternative sampling strategies such as LHS may lead to different convergence behaviors. In the next section, we examine how replacing CMC with LHS affects the performance of the estimators of interest.

\section{Effects of Latin Hypercube Sampling on Nested Simulation Estimators of the Variance of a Conditional Expectation}
\label{sec:ns_lhs}

Latin hypercube sampling (LHS) is a classical variance-reduction technique that replaces CMC sampling with a stratified design, often improving efficiency for MC estimators of high-dimensional integrals. In the context of Sobol' index estimation via nested simulation, performance depends on both the outer-level sampling scheme and the accuracy of inner-level approximations of conditional quantities. Because LHS imposes structure on the sample of outer-level scenarios, its interaction with nested simulation estimators is not automatic: the induced dependence can change the bias and variance behavior relative to CMC.

In this section, we examine Sobol' index estimators under LHS. Subsection~\ref{subsec:lhs-pfcr} studies the impact of LHS on the PF and CR estimators. Subsection~\ref{subsec:standard_nest} analyzes the effect of LHS on the NS estimator, while Subsection~\ref{subsec:bias_correct_nest} investigates its influence on several bias-reduction nested simulation estimators, namely the JK, SJ, and OH estimators.

\subsection{Effects of Latin Hypercube Sampling on the PF and CR Estimators}
\label{subsec:lhs-pfcr}
In this subsection, we study the PF estimator in~\eqref{eq:pick_freeze} and the CR estimator in~\eqref{eq:corr2} when the underlying sampling scheme is changed from CMC to LHS. The resulting estimators are referred to as the L--PF and L--CR estimators, respectively. We begin with the L--PF estimator.

Recall from Section~\ref{subsec:pick_freeze} that the PF estimator is constructed using $\cY = f(\bX_{\cu}, \bX_{-\cu})$ and $\cY^{\prime} = f(\bX_{\cu}, \bX_{-\cu}^{\prime})$, where $\bX_{-\cu}^{\prime}$ is an independent copy of $\bX_{-\cu}$. Under LHS, for a given outer-level sample size $K$, we first generate the outer-level scenarios $\{\lhs{\bX}_{\cu,i}\}_{i=1}^K$. Then, conditional on each $\lhs{\bX}_{\cu,i}$, we generate two inner-level input subvectors, $\lhs{\bX}_{-\cu,i}$ and $\lhs{\bX}^{\prime}_{-\cu,i}$, via LHS. The corresponding model outputs are $\cY_i = f(\lhs{\bX}_{\cu,i}, \lhs{\bX}_{-\cu,i})$ and $\cY_i^{\prime} = f(\lhs{\bX}_{\cu,i}, \lhs{\bX}^{\prime}_{-\cu,i})$. In practice, the aforementioned input subvectors are constructed from Latin hypercube design matrices. Specifically, we generate two $K\times p$ matrices, $\bA^{(1)}$ and $\bA^{(2)}$, using LHS, and define
\begin{equation}\label{eq:lhs_pf_x}
	\lhs{\bX}_{\cu,i} \coloneqq \bA_{i,\cu}^{(1)}, \qquad
	\lhs{\bX}_{-\cu,i} \coloneqq \bA_{i,-\cu}^{(1)}, \qquad
	\lhs{\bX}^{\prime}_{-\cu,i} \coloneqq \bA_{i,-\cu}^{(2)}, \qquad i \in [K] \ ,
\end{equation}
where $\bA_{i,\cu}^{(1)}$, $\bA_{i,-\cu}^{(1)}$, and $\bA_{i,-\cu}^{(2)}$ are the subvectors of the $i$th row of $\bA^{(1)}$ and $\bA^{(2)}$ corresponding to the column index sets  $\cu$ and $-\cu$, as appropriate.

The PF estimator under LHS (\textbf{L--PF}) retains the same form as~\eqref{eq:pick_freeze}, where the input combinations are constructed via LHS according to~\eqref{eq:lhs_pf_x}:
\begin{equation}\label{eq:pick_freeze_lhs}
	V_{\mbox{\tiny L--PF}}^{\cu}=\frac{1}{K} \sum_{i=1}^{K} f(\lhs{\bX}_{\cu,i}, \lhs{\bX}_{-\cu,i}) f(\lhs{\bX}_{\cu,i}, \lhs{\bX}_{-\cu,i}^{\prime})
	- \left(\frac{1}{2K} \sum_{i=1}^{K}\left[f(\lhs{\bX}_{\cu,i}, \lhs{\bX}_{-\cu,i})+f(\lhs{\bX}_{\cu,i}, \lhs{\bX}_{-\cu,i}^{\prime})\right]\right)^2 .
\end{equation}

The following result shows that, relative to the PF estimator under CMC, LHS does not improve the convergence rate of either the bias or the variance of the L--PF estimator. The proof is provided in Appendix~\ref{app:proof_pf_lhs_bias_variance}.
\begin{proposition}\label{prop:pf_lhs_bias_variance}
	The bias of $V_{\mbox{\tiny L--PF}}^{\cu}$ satisfies $\E{V_{\mbox{\tiny L--PF}}^{\cu}} - V = \cO(K^{-1})$, and its variance satisfies $\Var{V_{\mbox{\tiny L--PF}}^{\cu}} = \cO(K^{-1})$.
\end{proposition}
Proposition~\ref{prop:pf_lhs_bias_variance} implies that, given a total computational budget $T$, LHS preserves the $\cO(T^{-1})$ MSE convergence rate of the PF estimator under CMC.

We next analyze the \textbf{L--CR} estimator, i.e., the CR estimator constructed under LHS. As in the L--PF case, we first generate two independent $K\times p$ Latin hypercube design matrices, $\bA^{(1)}$ and $\bA^{(2)}$. We further generate a third $K\times p$ Latin hypercube design matrix $\bA^{(3)}$, independent of both $\bA^{(1)}$ and $\bA^{(2)}$, and define
\begin{equation}\label{eq:lhs_cr_x}
	\lhs{\bX}_{\cu,i}^{\prime\prime} \coloneqq \bA_{i,\cu}^{(3)}, \qquad
	\lhs{\bX}_{\cu,i}^{\prime} \coloneqq \bA_{i,\cu}^{(2)} \ .
\end{equation}

The L--CR estimator retains the same  form as~\eqref{eq:corr2}, where the input combinations are constructed via LHS according to \eqref{eq:lhs_pf_x} and \eqref{eq:lhs_cr_x}:
\begin{equation}\label{eq:corr2_lhs}
	V_{\mbox{\tiny L--CR}}^{\cu} = \frac{1}{K} \sum_{i=1}^K
	\left( f(\lhs{\bX}_{\cu,i}, \lhs{\bX}_{-\cu,i} ) - f(\lhs{\bX}_{\cu,i}^{\prime\prime}, \lhs{\bX}_{-\cu,i} ) \right) \cdot
	\left( f(\lhs{\bX}_{\cu,i}, \lhs{\bX}_{-\cu,i}^{\prime} ) - f( \lhs{\bX}_{\cu,i}^{\prime}, \lhs{\bX}_{-\cu,i}^{\prime} )\right) .
\end{equation}

The following result shows that, relative to the CR estimator under CMC, LHS does not improve the variance convergence rate of the L--CR estimator. The proof is provided in Appendix~\ref{app:proof_cr_lhs_bias_variance}.
\begin{proposition}\label{prop:cr_lhs_bias_variance}
	The estimator $V_{\mbox{\tiny L--CR}}^{\cu}$ is unbiased, and its variance satisfies $\Var{V_{\mbox{\tiny L--CR}}^{\cu}} = \cO(K^{-1})$.
\end{proposition}
Proposition~\ref{prop:cr_lhs_bias_variance} implies that, given a total computational budget $T$, the L--CR estimator retains the $\cO(T^{-1})$ MSE convergence rate under LHS.

Overall, the above analysis shows that, for both the PF and CR estimators, replacing CMC with LHS preserves the $\cO(T^{-1})$ MSE convergence rate, but does not lead to any further improvement. We next turn to nested simulation under LHS, where LHS can play a more significant role in improving the MSE convergence rate.

\subsection{Effects of Latin Hypercube Sampling on the Nested Simulation Estimator}
\label{subsec:standard_nest}
In this subsection, we analyze the standard nested simulation (NS) estimator under LHS, referred to as the \textbf{L--NS} estimator. Given an outer-level sample size $K$ and an inner-level sample size $N$, we construct the estimator in a manner analogous to the L--PF estimator, with LHS applied at both the outer and inner sampling levels. Specifically, we first generate the outer-level scenarios $\{\lhs{\bX}_{\cu,i}\}_{i=1}^K$ using LHS. Conditional on each $\lhs{\bX}_{\cu,i}$, we then generate $N$ inner-level realizations for $\bX_{-\cu}$ using LHS. In practice, we generate $N$ independent $K \times p$ Latin hypercube design matrices, denoted by $\bA^{(j)}$ for $j \in [N]$, and define
\[
\lhs{\bX}_{\cu,i} \coloneqq \bA_{i,\cu}^{(1)}, \qquad
\lhs{\bX}_{-\cu,j}^{(i)} \coloneqq \bA_{i,-\cu}^{(j)},
\qquad i \in [K],\ j \in [N] \ .
\]
For each input combination $(\lhs{\bX}_{\cu,i},\lhs{\bX}_{-\cu,j}^{(i)})$, we define the corresponding model output as
$
\lhs{\cY}_{ij} \coloneqq f(\lhs{\bX}_{\cu,i},\lhs{\bX}_{-\cu,j}^{(i)}).
$
The NS estimator based on LHS (L--NS) is then
\begin{equation}\label{eq:lhs_ns_est}
	V_{\mbox{\tiny L--NS}}^{\cu}
	= \frac{1}{K}\sum_{i=1}^{K}\left(\lhs{L}_N(\lhs{\bX}_{\cu,i}) - \frac{1}{K}\sum_{i=1}^{K}\lhs{L}_N(\lhs{\bX}_{\cu,i})\right)^2 ,
\end{equation}
where $\lhs{L}_N(\lhs{\bX}_{\cu,i}) \coloneqq {N}^{-1}\sum_{j=1}^{N}\lhs{\cY}_{ij}$ for each $i\in[K]$.

To facilitate the analysis of the impact of LHS, we define the key quantity
\begin{equation}\label{eq:remain_effect}
	R_\cu \coloneqq \frac{S_{T}^{-\cu} - \sum_{j \in -\cu}\bigl(S^{j} + S^{\cu,j}\bigr)}{S_{T}^{-\cu}} \ ,
\end{equation}
where
\[
S_{T}^{-\cu} \coloneqq \frac{\E{\Var{\cY \mid \bX_{\cu}}}}{\Var{\cY}}, \quad
S^{\cu,j} \coloneqq \frac{\Var{\E{\cY \mid \bX_{\cu \cup \{j\}}}} - \Var{\E{\cY \mid \bX_\cu}} - \Var{\E{\cY \mid X_j}}}{\Var{\cY}} \ ,
\]
and recall that $S^{\cu} \coloneqq \Var{\E{\cY \mid \bX_\cu}} / \Var{\cY}$. Here, $S_{T}^{-\cu}$ represents the total effect of $\bX_{-\cu}$, i.e., the sum of the first-order and higher-order interaction effects among the inputs in the complement set $-\cu$ \citep{saltelli2010variance}. To avoid confusion, we note that $S^{\cu,j}$ and $S^{\cu \cup \{j\}}$ represent different quantities: the former captures interaction effects between $\bX_\cu$ and $X_j$, whereas the latter is the first-order effect of $\bX_{\cu \cup \{j\}}$, defined as $\Var{\E{\cY \mid \bX_{\cu \cup \{j\}}}}/\Var{\cY}$.
We refer to the numerator in~\eqref{eq:remain_effect},
$
S_{T}^{-\cu} - \sum_{j \in -\cu}\bigl(S^{j} + S^{\cu,j}\bigr),
$
as the \textit{remaining interaction effect} of $S_{T}^{-\cu}$, since it quantifies the portion of interaction effects involving variables in $\bX_{-\cu}$ that is not captured by $\sum_{j\in-\cu} \bigl(S^{j} + S^{\cu,j}\bigr)$. It is immediate that the remaining interaction effect lies in $[0,S_{T}^{-\cu}]$, and hence $R_\cu\in[0,1]$. Moreover, when $S_{T}^{-\cu} \approx \sum_{j\in-\cu}(S^{j}+S^{\cu,j})$, i.e., the remaining interaction effect is negligible, $R_\cu$ is close to $0$.

The following result establishes the asymptotic bias of the L--NS estimator. The proof is provided in Appendix~\ref{app:proof_nest_lhs_bias}.
\begin{proposition}\label{prop:lhs_nest_bias}
	The bias of $V_{\mbox{\tiny L--NS}}^{\cu}$ satisfies
	$\E{V_{\mbox{\tiny L--NS}}^{\cu}} - V
		= N^{-1}\left(\Var{\cY} - V\right) R_\cu + r_N + \cO(N^{-2}) + \cO(K^{-1})$,
	where $r_N = o(N^{-1})$ and $r_N>0$ for finite $N$. Furthermore, if $|\cu|=1$, then
$
	\E{V_{\mbox{\tiny L--NS}}^{\cu}} - V
	= N^{-1}\left(\Var{\cY} - V\right) R_\cu + r_N + \cO(N^{-2}) + \cO(K^{-3}) + \cO(T^{-1})$.
\end{proposition}

Recall that $R_\cu\in[0,1]$. Consequently, under the allocation $K=\cO(T^{2/3})$ and $N=\cO(T^{1/3})$, the leading $N^{-1}$ bias term of $V_{\mbox{\tiny L--NS}}^{\cu}$ is no larger than that of the standard NS estimator $V_{\mbox{\tiny NS}}^{\cu}$ under CMC, whose bias is $(\Var{\cY}-V)/N$ as $K,N\rightarrow\infty$. Moreover, when $R_\cu\approx 0$---that is, the remaining interaction effect is negligible---the bias of the L--NS estimator can be substantially smaller than that of the standard NS estimator under CMC. The following example illustrates this bias reduction.
\vskip1ex
\noindent\textbf{Illustrative example.} Consider the Ishigami function, a three-dimensional test case described in Section~\ref{sec:experiment}. For $\cu=\{1\}$, we have
\begin{align*}
	R_{\single{1}}
	&= \frac{S_T^{\{2,3\}} - \sum_{j \in \{2,3\}}(S^{j} + S^{1,j})}{S_T^{\{2,3\}}}
	= \frac{S_T^{\{2,3\}} - \sum_{j \in \{2,3\}}(1-S_T^{j}-S^{1})}{S_T^{\{2,3\}}} \\
	&= \frac{1 - S^{1} - \sum_{j \in \{2,3\}}(1-S_T^{j}-S^{1})}{1 - S^{1}}
	= \frac{S^{1} + S_T^{2} + S_T^{3}-1}{1 - S^{1}}.
\end{align*}
The second equality above follows from $S^{1} + S^{2} + S^{1,2} + S_T^{3}= 1$ and $S^{1} + S^{3} + S^{1,3} + S_T^{2}= 1$. Similarly,
\[
R_{\single{2}} = \frac{S^{2} + S_T^{1} + S_T^{3}-1}{1 - S^{2}}, \qquad
R_{\single{3}} = \frac{S^{3} + S_T^{1} + S_T^{2}-1}{1 - S^{3}}.
\]
For the Ishigami function, $X_2$ has no interaction with other inputs, so $S_T^{2}=S^{2}$. Consequently, $S^{1} + S_T^{2} + S_T^{3}=S^{1}+S^{2}+S_T^{3}=1$, where the last equality uses $S^{1} + S^{2} + S^{1,2} + S_T^{3}=1$ and $S^{1,2}=0$. This implies $R_{\single{1}}=0$. Similarly, since $S_T^{2}=S^{2}$ and $S^{2,3}=0$, we have $S^{3}+S_T^{1}+S_T^{2}=S^{3}+S_T^{1}+S^{2}+S^{2,3}=1$, which yields $R_{\single{3}}=0$. In contrast, $R_{\single{2}}\neq 0$ because $X_1$ and $X_3$ exhibit an interaction effect. Since $R_{\single{1}}=R_{\single{3}}=0$, the leading $N^{-1}$ bias term vanishes for $\cu=\{1\}$ and $\cu=\{3\}$, and the bias becomes $o(N^{-1})+\cO(K^{-3})+\cO(T^{-1})$. For $\cu=\{2\}$, the bias is improved only by the constant factor $R_{\single{2}}<1$.

The use of LHS also affects the variance of the L--NS estimator, as established in the following result. The proof is provided in Appendix~\ref{app:proof_nest_lhs_var}.
\begin{proposition}\label{prop:lhs_nest_var}
	The variance of $V_{\mbox{\tiny L--NS}}^{\cu}$ satisfies $
	\Var{V_{\mbox{\tiny L--NS}}^{\cu}} = \cO(K^{-1})$.
	Additionally, if $|\cu|=1$, then
$
	\Var{V_{\mbox{\tiny L--NS}}^{\cu}} = \cO(K^{-3}) + \cO(T^{-1}) .
$
\end{proposition}

Proposition~\ref{prop:lhs_nest_var} shows that when $|\cu|=1$, the variance of $V_{\mbox{\tiny L--NS}}^{\cu}$ decays faster than that of the standard NS estimator $V_{\mbox{\tiny NS}}^{\cu}$ under CMC. Combining Propositions~\ref{prop:lhs_nest_bias} and~\ref{prop:lhs_nest_var}, it follows that, for $|\cu|=1$, the MSE of $V_{\mbox{\tiny L--NS}}^{\cu}$ improves upon that of $V_{\mbox{\tiny NS}}^{\cu}$ under the same budget allocation, namely the asymptotically optimal allocation for $V_{\mbox{\tiny NS}}^{\cu}$ with $K=\cO(T^{2/3})$ and $N=\cO(T^{1/3})$. In this case, the MSE rate of $V_{\mbox{\tiny L--NS}}^{\cu}$ becomes $o(T^{-2/3})$, which is faster than the $\cO(T^{-2/3})$ rate achieved by $V_{\mbox{\tiny NS}}^{\cu}$. When $|\cu|>1$, applying LHS does not improve the MSE convergence rate.

The improvement for $|\cu|=1$ can be attributed to the stratification property of LHS, which reduces the variability of the outer-level sampling \citep{owen2013mcbook}. When $\cu$ contains a single input variable, LHS enforces uniform coverage of its marginal distribution, thereby diminishing the outer-level sampling variance contribution to the overall MSE. In contrast, when $|\cu|>1$, stratification is only marginal and becomes less effective for controlling variability over higher-dimensional subspaces. Consequently, the variance-reduction effect weakens, and the MSE convergence rate of the L--NS estimator remains comparable to that of the standard NS estimator under CMC.

In summary, incorporating LHS into the nested simulation framework improves efficiency primarily through variance reduction. This enhancement is most pronounced for $|\cu|=1$, where $V_{\mbox{\tiny L--NS}}^{\cu}$ achieves a faster MSE convergence rate than its CMC counterpart. However, this advantage diminishes as the dimensionality of $\cu$ increases. These theoretical insights suggest that LHS is particularly beneficial for estimating first-order Sobol' indices, a conclusion that is further supported by the numerical experiments in Section~\ref{sec:experiment}.

\subsection{Effects of Latin Hypercube Sampling on Bias-Reduced Nested Simulation Estimators}
\label{subsec:bias_correct_nest}

Recall from Subsection~\ref{subsec:jackknife} that, under CMC, the NS estimator $V_{\mbox{\tiny NS}}^{\cu}$ can be improved via bias-reduction techniques to achieve unbiasedness or faster MSE convergence rates. In this subsection, we analyze the impact of LHS on the bias-reduced NS estimators introduced in Subsection~\ref{subsec:jackknife}, namely, the JK estimator $V_{\mbox{\tiny JK}}^{\cu}$, the SJ estimator $V_{\mbox{\tiny SJ}}^{\cu}$, and the OH estimator $V_{\mbox{\tiny OH}}^{\cu}$. Our goal is to determine whether the stratification effect of LHS can further improve the convergence properties of these estimators beyond what is obtained under CMC.

Denote the corresponding estimators under LHS by the \textbf{L--SJ} estimator $V_{\mbox{\tiny L--SJ}}^{\cu}$, the \textbf{L--OH} estimator $V_{\mbox{\tiny L--OH}}^{\cu}$, and the \textbf{L--JK} estimator $V_{\mbox{\tiny L--JK}}^{\cu}$. We begin by presenting the bias of the L--SJ and L--OH estimators, with proofs given in Appendices~\ref{app:proof_sj_lhs_bias} and~\ref{app:proof_oh_lhs_bias}, respectively.
\begin{proposition}\label{prop:bias_sj_lhs}
	The bias of $V_{\mbox{\tiny L--SJ}}^{\cu}$ satisfies
	$
	\E{V_{\mbox{\tiny L--SJ}}^{\cu}}- V =  \cO(J^{-1}) + r_N + \cO(N^{-2}) ,
	$
	where $r_N = o(N^{-1})$ and $r_N>0$ for finite $N$.
\end{proposition}

Proposition~\ref{prop:bias_sj_lhs} shows that, under LHS, the L--SJ estimator retains a bias of order $\cO(J^{-1}) + r_N + \cO(N^{-2})$. Under CMC, the bias-reduction terms $(L_{N,-l}(\bX_{\cu,i})-\widehat{\mu})^2$ given in \eqref{eq:fse} are constructed from conditionally independent outer-level scenarios, which enables the bias-cancellation argument. In contrast, LHS induces dependence among the outer-level scenarios, so the same cancellation no longer applies. Nevertheless, the bias in Proposition~\ref{prop:bias_sj_lhs} vanishes as $J, N\rightarrow\infty$.

\begin{proposition}\label{prop:bias_oh_lhs}
	The bias of $V_{\mbox{\tiny L--OH}}^{\cu}$ satisfies
	$
	\E{V_{\mbox{\tiny L--OH}}^{\cu}}- V =  \cO(K^{-1}) + r_N + \cO(N^{-2}) ,
	$
	where $r_N = o(N^{-1})$ and $r_N>0$ for finite $N$.
\end{proposition}
Proposition~\ref{prop:bias_oh_lhs} shows that the L--OH estimator also exhibits a bias of order $\cO(K^{-1}) + r_N + \cO(N^{-2})$ under LHS. Although the OH estimator is unbiased under CMC by construction, dependence among stratified outer-level scenarios under LHS prevents exact cancellation.

For both the L--SJ and L--OH estimators, the respective biases vanish only in the limit as the inner-level sample size $N\to\infty$. Therefore, if one adopts the same allocation strategy commonly used for the SJ and OH estimators under CMC, namely, fixing $N$, the bias terms do not vanish as $T\to\infty$. In particular, because $r_N>0$ for finite $N$, the biases of $V_{\mbox{\tiny L--SJ}}^{\cu}$ and $V_{\mbox{\tiny L--OH}}^{\cu}$ remain $\Theta(1)$, dominate the MSE, and hence lead to an $\cO(1)$ MSE convergence rate.

We next present the bias and variance of the L--JK estimator, with the corresponding proofs provided in Appendices~\ref{app:proof_jk_lhs_bias} and~\ref{app:proof_jk_lhs_var}, respectively.

\begin{proposition}\label{prop:bias_jack_lhs}
	The bias of $V_{\mbox{\tiny L--JK}}^{\cu}$ satisfies
	$
	\E{V_{\mbox{\tiny L--JK}}^{\cu}}- V =  \cO(K^{-1}) + r_N + \cO(N^{-2}) ,
	$
	where $r_N = o(N^{-1})$ and $r_N>0$ for finite $N$.
\end{proposition}

Proposition~\ref{prop:bias_jack_lhs} shows that the JK estimator becomes biased under LHS. As with the L--SJ and L--OH estimators, this bias is induced by dependence among stratified outer-level scenarios. To characterize the MSE behavior of the L--JK estimator under the same budget allocation used for the JK estimator under CMC, we next establish its variance convergence rate.

\begin{proposition}\label{prop:var_jack_lhs}
	The variance of $V_{\mbox{\tiny L--JK}}^{\cu}$ satisfies
	$
	\Var{V_{\mbox{\tiny L--JK}}^{\cu}} = \cO(K^{-1}) + \cO(N^{-2}) .
	$
\end{proposition}

Combining Propositions~\ref{prop:bias_jack_lhs} and~\ref{prop:var_jack_lhs}, we obtain that the MSE of $V_{\mbox{\tiny L--JK}}^{\cu}$ remains $\cO(T^{-2/3})$ under the allocation $K=\cO(T^{2/3})$ and $N=\cO(T^{1/3})$.

In summary, while LHS can reduce variance through stratification, it also induces dependence among outer-level scenarios, thereby altering the bias-cancellation properties of bias-reduced nested simulation estimators. Consequently, the L--JK, L--SJ, and L--OH estimators generally exhibit residual bias under LHS. To prevent this bias from dominating the MSE, the inner-level sample size $N$ must increase with the total simulation budget $T$. These results highlight the need for caution when combining LHS with bias-reduction constructions in nested simulation, especially when using budget allocation strategies developed under CMC that keep $N$ fixed.

We conclude this section with a few remarks. The benefit of LHS depends strongly on the estimator structure. For the PF and CR estimators, LHS does not change the asymptotic bias and variance orders and therefore does not improve the $\cO(T^{-1})$ MSE convergence rate. For the standard NS estimator, LHS can be advantageous, especially for first-order indices ($|\cu|=1$), where stratification reduces the contribution of outer-level sampling variance to the MSE and yields a faster convergence rate than the CMC-based NS estimator. In contrast, when combined with bias-reduction constructions (JK, SJ, and OH), LHS induces dependence across outer-level scenarios that prevents exact bias cancellation. Consequently, to retain the intended benefits of bias correction under LHS, the inner-level sample size $N$ must grow with the total budget $T$,  rather than adopting CMC-based allocation strategies with fixed $N$.

\section{Numerical Evaluations}\label{sec:experiment}

In this section, we conduct numerical experiments to compare the performance of the estimators of the Sobol' index numerator under CMC and LHS, as described in Subsections \ref{subsec:numerical_SRS} and \ref{subsec:numerical_LHS}, respectively. The methods considered include the pick-freeze (PF) estimator, the Correlation 2 (CR) estimator, the nested simulation (NS) estimator, the $1\frac{1}{2}$-level nested-simulation (OH) estimator, the unbiased jackknife (JK) estimator, and the split jackknife (SJ) estimator. Recall that the PF estimator is a widely used approach for estimating Sobol' indices, whereas the CR estimator is particularly effective when the Sobol' index is small. Although the OH estimator has demonstrated strong performance in financial risk management applications, it has not been previously examined in the context of Sobol' index estimation. This study therefore provides an opportunity to assess its applicability and comparative efficiency in GSA.

We note that the theoretical results in Sections~\ref{sec:nest} and \ref{sec:ns_lhs} concern estimators of the Sobol' index numerator, $V=\Var{\E{\cY\mid \bX_\cu}}$. In this section, however, we report the MSE of the corresponding Sobol' index estimator, since the normalized index is the primary quantity of interest in GSA. The denominator, $\Var{\cY}$, is estimated from the simulation outputs using the standard sample-variance estimator, whose MSE is of order $\cO(T^{-1})$. Therefore, the denominator estimation error does not alter the convergence-rate comparisons suggested by the numerator analysis. When the numerator estimator converges more slowly than $\cO(T^{-1})$, the numerator error dominates the ratio error; when the numerator estimator has MSE of order $\cO(T^{-1})$, the denominator contributes at the same order. Hence, the Sobol' index MSEs reported below should be interpreted as empirical evidence supporting the numerator-based theory, up to constants introduced by ratio normalization.

We aim to estimate the first-order Sobol' index $S^{\cu}$ for $\cu=\{i\}$ with each $i\in[p]$, using a total budget $T \in \{0.5 \times 10^4, 1 \times 10^4, 0.5 \times 10^5, 1 \times 10^5, 0.5 \times 10^6, 1 \times 10^6 \}$. A fraction $(p+1)^{-1}T$ is allocated to estimate the common denominator $\Var{\cY}$ for all Sobol' indices, and the remaining budget is allocated equally across the $p$ input dimensions to estimate $\Var{\E{\cY \mid X_i}}$ for $i\in[p]$.
We adopt the following budget allocation strategies for constructing different estimators under CMC and apply the same strategies under LHS. For the OH and NS estimators, which require a pilot experiment to estimate unknown quantities, we allocate $10\%$ of the total budget $T$ to each input dimension for this purpose. For the JK estimator, we set $K = \left\lceil T^{2/3} \right\rceil$ and $N = \left\lceil T/K \right\rceil$. For the SJ estimator, we allocate $10\%$ of the total budget to the preliminary dataset $\cD_{\mbox{\tiny pre}}$ and fix the inner-level sample size for estimating $V$ at $N = 10$. Our numerical results indicate that $\cD_{\mbox{\tiny pre}}$ need not be large and that the estimator's performance is robust to the choice of $N$.
For performance evaluation, we conduct independent macro-replications and use the MSE as the performance metric, defined as
$
\mbox{MSE} \coloneqq M^{-1}\sum_{m=1}^{M}\left(\widehat{S}^{i}_{m} - S^{i} \right)^2,
$
where $\widehat{S}^{i}_{m}$ denotes the estimator of $S^{i}$ obtained from the $m$th macro-replication, computed as the ratio of the chosen estimator of $\Var{\E{\cY \mid X_i}}$ to the estimate of $\Var{\cY}$. In all experiments, we take $M=1{,}000$.

We consider three numerical examples. The first two are benchmark test functions commonly used in the GSA literature, whereas the third is a practical case study drawn from a real-world application. The true Sobol' indices for all examples are reported in Appendix~\ref{app:add_info_numerical}.

\vskip1ex
\noindent\textbf{\emph{Ishigami function.}} The Ishigami function is a classical benchmark model for evaluating the performance of Sobol' index estimators \citep{ishigami1990importance}. It is defined as
$
\cY=\sin \left(X_{1}\right)+7 \sin^2 \left(X_{2}\right)+0.1 X_{3}^{4} \sin \left(X_{1}\right),
$
where the $X_i$'s are independent random variables uniformly distributed on $[-\pi,\pi]$.

\vskip1ex
\noindent\textbf{\emph{$g$-function.}} The $g$-function is widely used for assessing the performance of Sobol' index estimators \citep{owen2013better}. A $p$-dimensional $g$-function is defined as
$
\cY = \prod_{i=1}^{p} g_i(X_i),
$
where $g_i(X_i) = (|4X_i - 2| + a_i)/(1 + a_i)$ with $a_i \ge 0$, and the $X_i$'s are independent and uniformly distributed over $[0,1]$. We consider two cases: $p=3$ and $p=5$. In the three-dimensional (3D) case, $a_1=19$, $a_2=9$, and $a_3=4$; in the five-dimensional (5D) case, $a_i=i$ for $i\in[5]$.

\vskip1ex
\noindent\textbf{\emph{Hydrological model (hymod).}} This example simulates rainfall--runoff processes and is widely used in hydrological modeling and sensitivity analysis \citep{wagener2001framework}. Its structure consists of a nonlinear tank connected to three parallel linear tanks that represent surface flow, together with a separate slow-flow tank that models groundwater movement. The model depends on five key input parameters: {\tt Sm} ($X_1$, maximum watershed storage capacity), {\tt beta} ($X_2$, spatial variability of soil moisture capacity), {\tt alpha} ($X_3$, flow partitioning coefficient between fast and slow pathways), and two residence-time parameters, {\tt Rf} and {\tt Rs} ($X_4$ and $X_5$, representing the quick-flow and slow-flow tanks, respectively). Detailed descriptions and distributions of these parameters are provided in Appendix~\ref{app:add_info_numerical}. The model output in this study is the Nash--Sutcliffe efficiency (NSE), a widely used performance measure in hydrology that quantifies the agreement between simulated and observed streamflow. All hymod simulations are implemented using the MATLAB package developed by \cite{pianosi2015matlab}.

\subsection{Numerical Evaluation of Sobol' Index Estimators under CMC}\label{subsec:numerical_SRS}

 This subsection examines the performance of the Sobol' index estimators under CMC. Figure \ref{fig:ishigami_mse} reports the estimated MSE convergence rates for all estimators applied to the Ishigami function example. Two reference lines with slopes $-1$ and $-2/3$ are included for comparison. The results show that the MSEs of the NS and JK estimators converge at a rate of $\cO(T^{-2/3})$, whereas the remaining estimators achieve the faster rate $\cO(T^{-1})$, consistent with the theoretical results.
Several general observations emerge. The NS estimator consistently exhibits inferior performance across all input variables. The SJ and OH estimators display very similar performance. The PF estimator performs best for large Sobol' indices (e.g., $S^{1}=0.3134$ and $S^{2}=0.4424$) but performs poorly when the index is small. Conversely, the CR estimator is effective for small indices but less accurate for large ones. In contrast, the JK, OH, and SJ estimators produce stable and accurate estimates across a wide range of index magnitudes, demonstrating robustness to index size.

\begin{figure}[!htbp]
     \centering
     \begin{subfigure}[b]{0.48\textwidth}
         \centering
         \includegraphics[width=\textwidth]{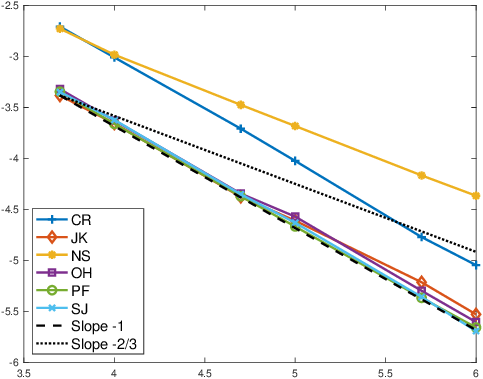}
         \caption{$X_1$}
     \end{subfigure}
     \hfill
     \begin{subfigure}[b]{0.48\textwidth}
         \centering
         \includegraphics[width=\textwidth]{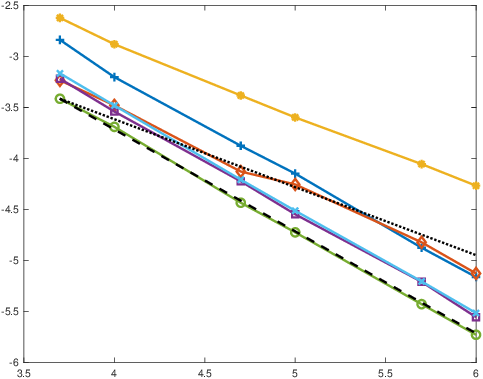}
         \caption{$X_2$}
     \end{subfigure}
     \hfill
     \begin{subfigure}[b]{0.48\textwidth}
         \centering
         \includegraphics[width=\textwidth]{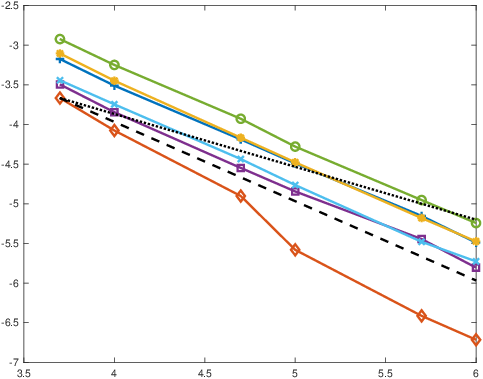}
         \caption{$X_3$}
     \end{subfigure}
\caption{Log--log MSE under CMC (y-axis)  versus total budget $T$ (x-axis)  for the Ishigami function example. Reference slopes $-1$ and $-2/3$ indicate rates $\cO(T^{-1})$ and $\cO(T^{-2/3})$.}
\label{fig:ishigami_mse}
\end{figure}

Figures \ref{fig:gfunc_3d_mse} and \ref{fig:gfunc_5d_mse} report the MSE convergence rates for the $g$-function example and show trends similar to those observed for the Ishigami function example. Specifically, for $X_3$ in the 3D case (and $X_1$ in the 5D case), both of which correspond to relatively large Sobol' indices, the PF estimator performs best, followed by the SJ and OH estimators, whereas the JK, CR, and NS estimators perform worse.
For $X_1$ and $X_2$ in the 3D case (and $X_2$---$X_5$ in the 5D case), which are associated with small Sobol' indices, the CR estimator achieves the best performance, and the SJ and OH estimators remain competitive. The JK estimator performs well for small budgets, but its relative efficiency deteriorates as the budget increases due to its slower MSE convergence rate. The PF and NS estimators consistently exhibit the poorest performance when the Sobol' index is small.

\begin{figure}[!htbp]
     \centering
     \begin{subfigure}[b]{0.48\textwidth}
         \centering
         \includegraphics[width=\textwidth]{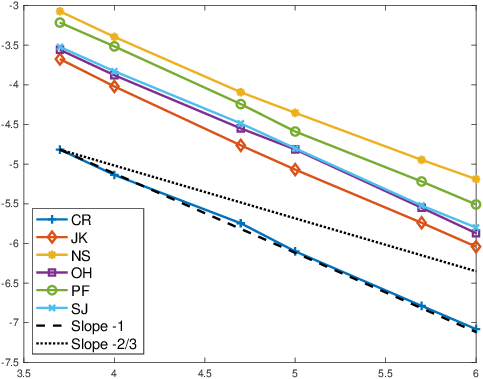}
         \caption{$X_1$}
     \end{subfigure}
     \hfill
     \begin{subfigure}[b]{0.48\textwidth}
         \centering
         \includegraphics[width=\textwidth]{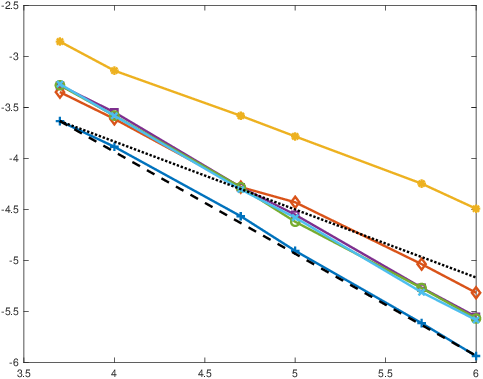}
         \caption{$X_2$}
     \end{subfigure}
     \hfill
     \begin{subfigure}[b]{0.48\textwidth}
         \centering
         \includegraphics[width=\textwidth]{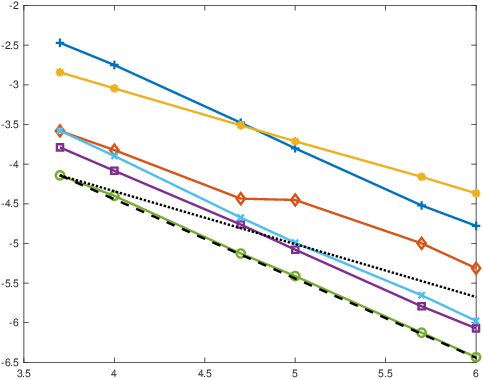}
         \caption{$X_3$}
     \end{subfigure}
\caption{Log--log MSE under CMC (y-axis)  versus total budget $T$ (x-axis)  for the $3$D  $g$-function example. Reference slopes $-1$ and $-2/3$ indicate rates $\cO(T^{-1})$ and $\cO(T^{-2/3})$.}
        \label{fig:gfunc_3d_mse}
\end{figure}

\begin{figure}[!htbp]
     \centering
     \begin{subfigure}[b]{0.48\textwidth}
         \centering
         \includegraphics[width=\textwidth]{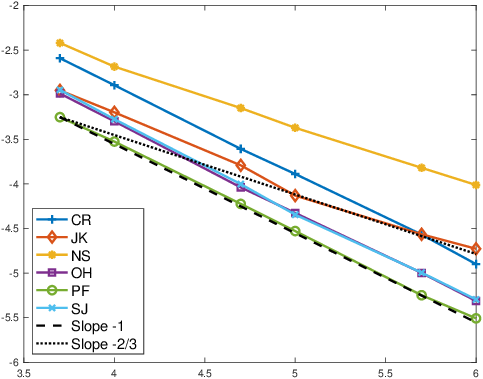}
         \caption{$X_1$}
     \end{subfigure}
     \hfill
     \begin{subfigure}[b]{0.48\textwidth}
         \centering
         \includegraphics[width=\textwidth]{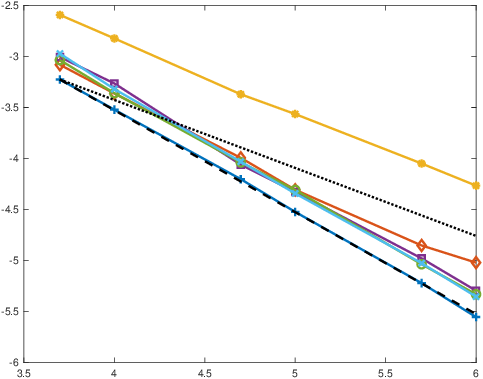}
         \caption{$X_2$}
     \end{subfigure}
     \hfill
     \begin{subfigure}[b]{0.48\textwidth}
         \centering
         \includegraphics[width=\textwidth]{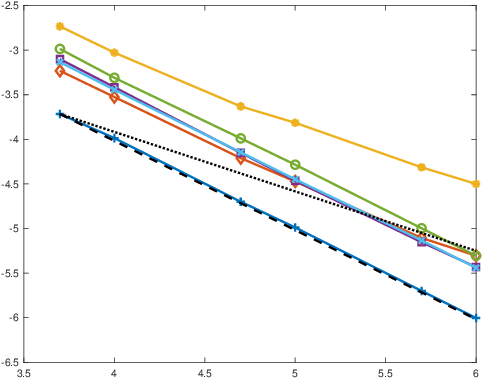}
         \caption{$X_3$}
     \end{subfigure}
     \hfill
     \begin{subfigure}[b]{0.48\textwidth}
         \centering
         \includegraphics[width=\textwidth]{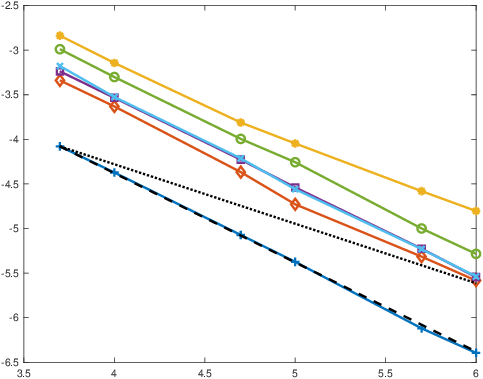}
         \caption{$X_4$}
     \end{subfigure}
     \hspace{0.15\textwidth}
     \begin{subfigure}[b]{0.48\textwidth}
         \centering
         \includegraphics[width=\textwidth]{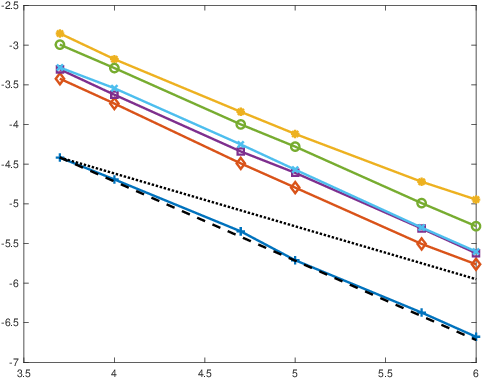}
         \caption{$X_5$}
     \end{subfigure}
\caption{Log--log MSE under CMC (y-axis)  versus total budget $T$ (x-axis)  for the $5$D  $g$-function example. Reference slopes $-1$ and $-2/3$ indicate rates $\cO(T^{-1})$ and $\cO(T^{-2/3})$.}
        \label{fig:gfunc_5d_mse}
\end{figure}

Figure \ref{fig:hymod_mse} illustrates the performance of the various estimators for the hymod example and reveals patterns consistent with those observed in the previous two test cases. For inputs $X_3$ and $X_5$, which correspond to relatively large Sobol' indices, most estimators perform well, with the notable exceptions of the CR and NS estimators. In contrast, for inputs $X_1$, $X_2$, and $X_4$, which are associated with small Sobol' indices, the CR estimator achieves the best performance, whereas the PF and NS estimators yield the least favorable results.

\begin{figure}[!htbp]
     \centering
     \begin{subfigure}[b]{0.48\textwidth}
         \centering
         \includegraphics[width=\textwidth]{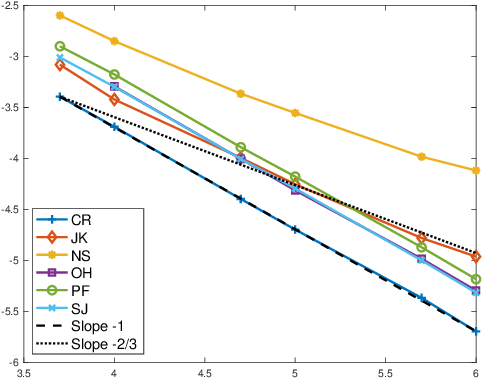}
         \caption{$X_1$}
     \end{subfigure}
     \hfill
     \begin{subfigure}[b]{0.48\textwidth}
         \centering
         \includegraphics[width=\textwidth]{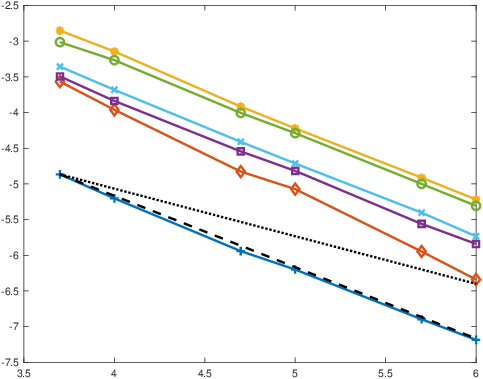}
         \caption{$X_2$}
     \end{subfigure}
     \hfill
     \begin{subfigure}[b]{0.48\textwidth}
         \centering
         \includegraphics[width=\textwidth]{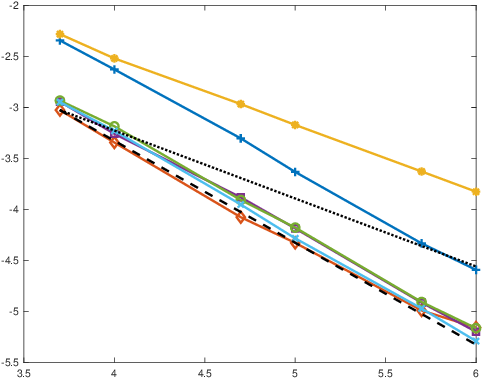}
         \caption{$X_3$}
     \end{subfigure}
	 \hfill
     \vspace{0.5em}
     \begin{subfigure}[b]{0.48\textwidth}
         \centering
         \includegraphics[width=\textwidth]{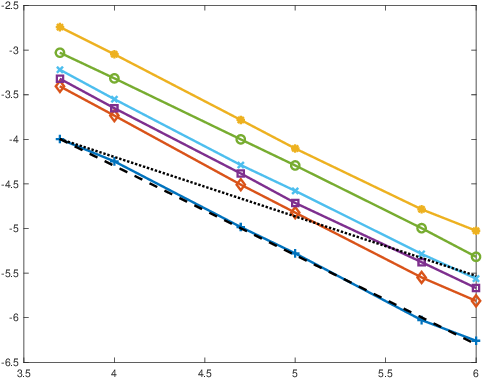}
         \caption{$X_4$}
     \end{subfigure}
     \hspace{0.15\textwidth}
     \begin{subfigure}[b]{0.48\textwidth}
         \centering
         \includegraphics[width=\textwidth]{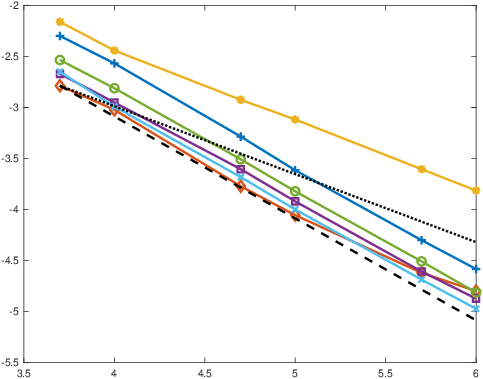}
         \caption{$X_5$}
     \end{subfigure}

\caption{Log--log MSE under CMC (y-axis)  versus total budget $T$ (x-axis)  for the hymod example. Reference slopes $-1$ and $-2/3$ indicate rates $\cO(T^{-1})$ and $\cO(T^{-2/3})$.}
     \label{fig:hymod_mse}
\end{figure}

We conclude this subsection with Table \ref{tab:performance}, which summarizes the findings and provides practical recommendations for selecting estimators of the Sobol' index numerator under CMC. When the Sobol' index is expected to be small, the CR estimator is the most suitable. Conversely, when the index is anticipated to be large---specifically, greater than $0.4$---the PF estimator is preferred. When no prior information about the index magnitude is available, the proposed jackknife estimators (JK in \eqref{eq:jack} and SJ in \eqref{eq:fse}), together with the OH estimator, provide robust performance across a broad range of index values. Although the SJ and OH estimators exhibit comparable overall accuracy, the SJ estimator is generally preferred because it does not require a pilot experiment to estimate higher-order moments, which is needed to construct the OH estimator.

\begin{table}[t]
\centering
\caption{Performance summary and recommendations for estimators of the Sobol' index numerator under CMC.}
\label{tab:performance}
\begin{tabular}{cccp{4cm}p{5.5cm}}
\hline
Estimator &  Biased for $V$ &  MSE  & Applicability & Recommendations \\ \hline

PF in \eqref{eq:pick_freeze}
& Yes
& $\mathcal{O}(T^{-1})$
& large indices ($>0.4$)
& Preferred when the Sobol' indices are large \\

CR in \eqref{eq:corr2}
& No
& $\mathcal{O}(T^{-1})$
& small indices ($<0.1$)
& Best choice for small indices \\
OH in \eqref{eq:one-half}
& No
& $\mathcal{O}(T^{-1})$
& all ranges
& Robust across all index values; requires a pilot experiment \\

NS in \eqref{eq:nest}
& Yes
& $\mathcal{O}(T^{-2/3})$
& all ranges
& Robust across all index values; slower convergence \\

JK in \eqref{eq:jack}
& No
& $\mathcal{O}(T^{-2/3})$
& all ranges
& Robust across all index values; a good choice under a limited budget \\

SJ in \eqref{eq:fse}
& Yes
& $\mathcal{O}(T^{-1})$
& all ranges
& Robust and efficient; preferred over OH (no pilot experiment required) \\ \hline

\end{tabular}
\end{table}

\subsection{Numerical Evaluation of Sobol' Index Estimators under LHS}
\label{subsec:numerical_LHS}

This subsection compares the performance of the Sobol' index estimators under LHS. We apply the same budget allocation strategies as those used for the corresponding estimators under CMC.

Figure \ref{fig:ishigami_mse_lhs} reports the MSE convergence rates for all estimators for the Ishigami function example under LHS. The MSE of the L--NS estimator ranges from $\cO(T^{-2/3})$ to $\cO(T^{-1})$, and even improves beyond $\cO(T^{-1})$ when estimating $S^{3}$. The L--JK estimator roughly follows the $\cO(T^{-2/3})$ rate. In contrast, the MSEs of the L--SJ and L--OH estimators remain at $\cO(1)$ as the budget increases, because they employ fixed inner-level sample sizes; this behavior is consistent with Propositions \ref{prop:bias_sj_lhs} and \ref{prop:bias_oh_lhs}.
Several further observations are worth noting. The L--NS estimator substantially outperforms the others when estimating $S^{1}$ and $S^{3}$, consistent with the theoretical discussion in the illustrative example following Proposition \ref{prop:lhs_nest_bias}. Its improvement for $S^{2}$ is less pronounced, and the MSE convergence rate remains $\cO(T^{-2/3})$. The L--PF and L--CR estimators retain their respective strengths for estimating large and small Sobol' indices, as observed under CMC, with modest improvements under LHS.
The L--JK estimator performs poorly because it is no longer unbiased under LHS. Finally, the L--OH and L--SJ estimators exhibit the least favorable performance across all Sobol' indices, as their biases persist due to the fixed inner-level sample sizes.

\begin{figure}[!htbp]
     \centering
     \begin{subfigure}[b]{0.46\textwidth}
         \centering
         \includegraphics[width=\textwidth]{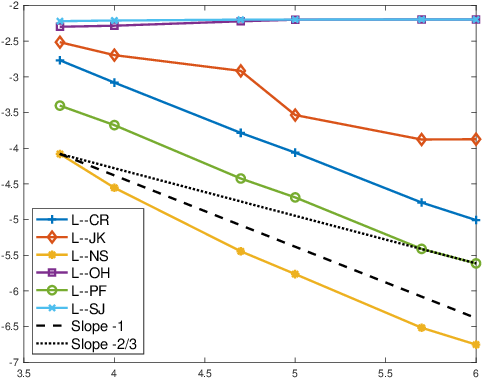}
         \caption{$X_1$}
     \end{subfigure}
     \hfill
     \begin{subfigure}[b]{0.46\textwidth}
         \centering
         \includegraphics[width=\textwidth]{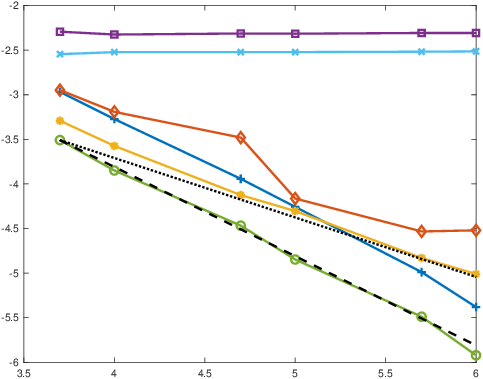}
         \caption{$X_2$}
     \end{subfigure}
     \hfill
     \begin{subfigure}[b]{0.46\textwidth}
         \centering
         \includegraphics[width=\textwidth]{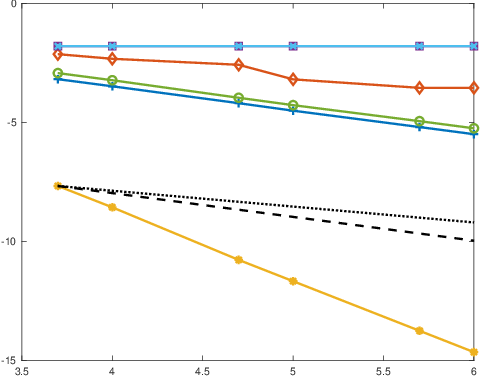}
         \caption{$X_3$}
     \end{subfigure}
\caption{Log--log MSE under LHS (y-axis) versus total budget $T$ (x-axis) for the Ishigami function example. Reference slopes $-1$ and $-2/3$ indicate rates $\cO(T^{-1})$ and $\cO(T^{-2/3})$.}
        \label{fig:ishigami_mse_lhs}
\end{figure}

Figures \ref{fig:gfunc_3d_mse_lhs} and \ref{fig:gfunc_5d_mse_lhs} show the MSE convergence rates for the $g$-function example under LHS and exhibit trends similar to those observed for the Ishigami function example. The L--NS estimator outperforms all other estimators across all input variables, whereas the L--JK, L--SJ, and L--OH estimators consistently exhibit inferior performance. The L--PF and L--CR estimators behave similarly to their CMC counterparts, with their relative performance depending on the magnitude of the Sobol' indices being estimated.

\begin{figure}[!htbp]
     \centering
     \begin{subfigure}[b]{0.46\textwidth}
         \centering
         \includegraphics[width=\textwidth]{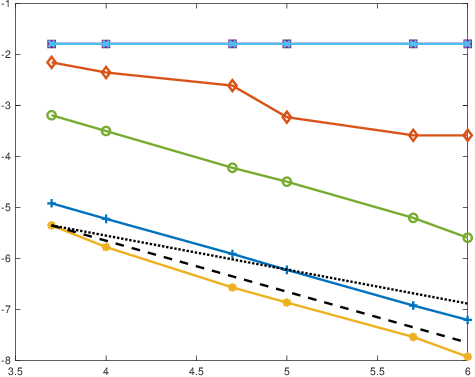}
         \caption{$X_1$}
     \end{subfigure}
     \hfill
     \begin{subfigure}[b]{0.46\textwidth}
         \centering
         \includegraphics[width=\textwidth]{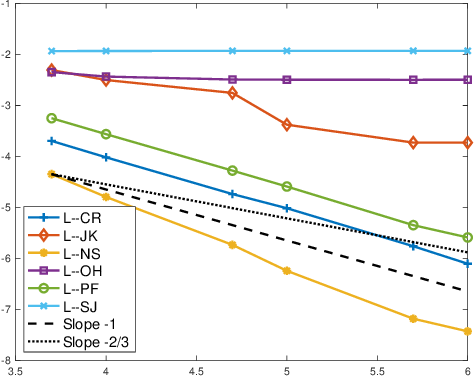}
         \caption{$X_2$}
     \end{subfigure}
     \hfill
     \begin{subfigure}[b]{0.46\textwidth}
         \centering
         \includegraphics[width=\textwidth]{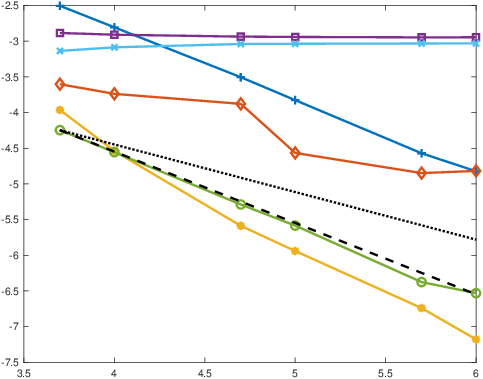}
         \caption{$X_3$}
     \end{subfigure}
\caption{Log--log MSE under LHS (y-axis) versus total budget $T$ (x-axis) for the 3D $g$-function example. Reference slopes $-1$ and $-2/3$ indicate rates $\cO(T^{-1})$ and $\cO(T^{-2/3})$.}
        \label{fig:gfunc_3d_mse_lhs}
\end{figure}

\begin{figure}[!htbp]
     \centering
     \begin{subfigure}[b]{0.46\textwidth}
         \centering
         \includegraphics[width=\textwidth]{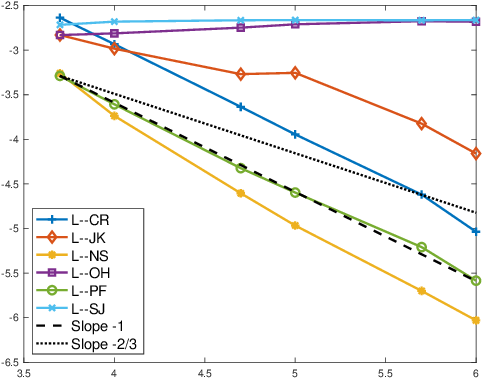}
         \caption{$X_1$}
     \end{subfigure}
     \hfill
     \begin{subfigure}[b]{0.46\textwidth}
         \centering
         \includegraphics[width=\textwidth]{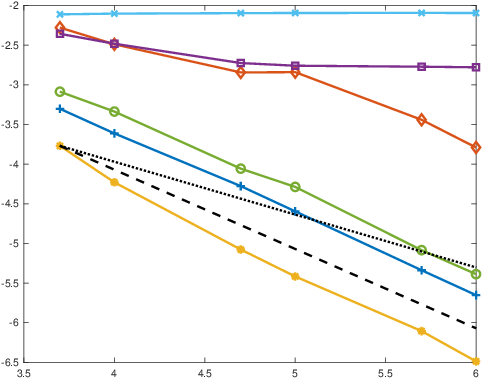}
         \caption{$X_2$}
     \end{subfigure}
     \hfill
     \begin{subfigure}[b]{0.46\textwidth}
         \centering
         \includegraphics[width=\textwidth]{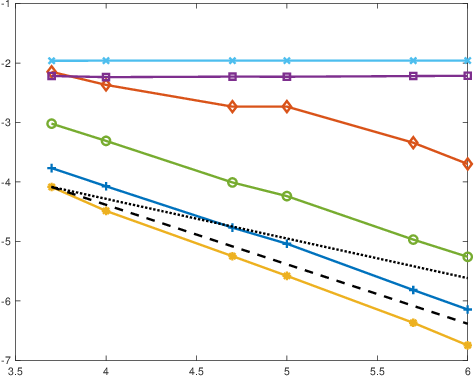}
         \caption{$X_3$}
     \end{subfigure}
     \hfill
     \begin{subfigure}[b]{0.46\textwidth}
         \centering
         \includegraphics[width=\textwidth]{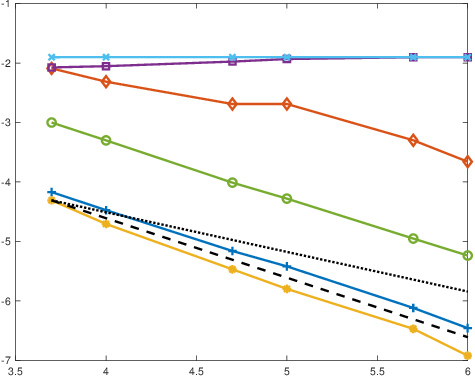}
         \caption{$X_4$}
     \end{subfigure}
     \hspace{0.15\textwidth}
     \begin{subfigure}[b]{0.46\textwidth}
         \centering
         \includegraphics[width=\textwidth]{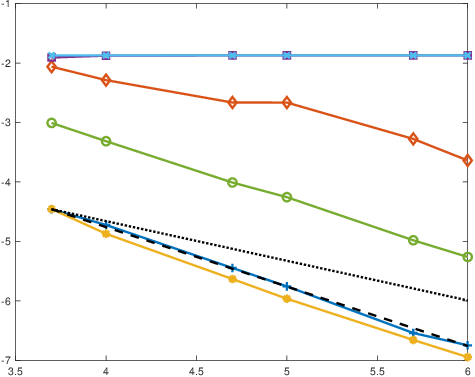}
         \caption{$X_5$}
     \end{subfigure}
\caption{Log--log MSE under LHS (y-axis) versus total budget $T$ (x-axis) for the 5D $g$-function example. Reference slopes $-1$ and $-2/3$ indicate rates $\cO(T^{-1})$ and $\cO(T^{-2/3})$.}
        \label{fig:gfunc_5d_mse_lhs}
\end{figure}

Figure \ref{fig:hymod_mse_lhs} illustrates the performance of various estimators for the hymod example under LHS and shows patterns largely consistent with those observed in the previous two examples, with minor differences. For inputs $X_3$ and $X_5$, which correspond to large Sobol' indices, the L--NS estimator performs best, followed by L--PF, whereas the remaining estimators are less competitive. Conversely, for inputs $X_1$, $X_2$, and $X_4$, which are associated with small Sobol' indices, the L--CR estimator performs best, followed by the L--NS estimator, while the other estimators yield the least favorable results. In this example, the L--NS estimator does not exhibit the dominant performance observed in the previous examples, largely because the non-negligible value of $R_{\cu}$ (defined in \eqref{eq:remain_effect}) for all inputs limits the improvement attainable from LHS.

\begin{figure}[!htbp]
     \centering
     \begin{subfigure}[b]{0.46\textwidth}
         \centering
         \includegraphics[width=\textwidth]{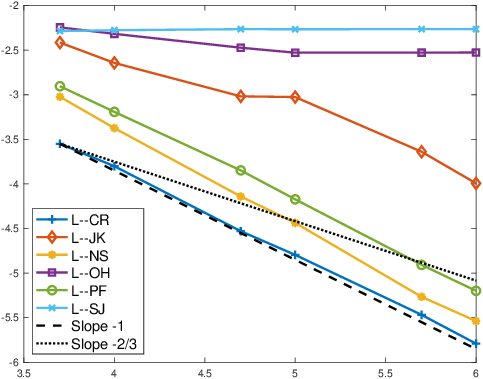}
         \caption{$X_1$}
     \end{subfigure}
     \hfill
     \begin{subfigure}[b]{0.46\textwidth}
         \centering
         \includegraphics[width=\textwidth]{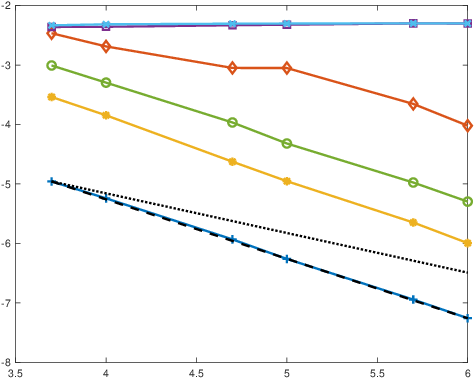}
         \caption{$X_2$}
     \end{subfigure}
     \hfill
     \begin{subfigure}[b]{0.46\textwidth}
         \centering
         \includegraphics[width=\textwidth]{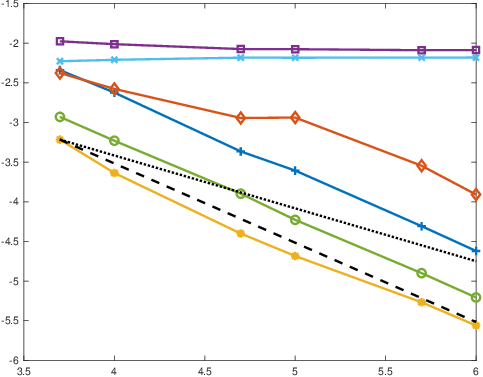}
         \caption{$X_3$}
     \end{subfigure}
	 \hfill
     \vspace{0.5em}
     \begin{subfigure}[b]{0.46\textwidth}
         \centering
         \includegraphics[width=\textwidth]{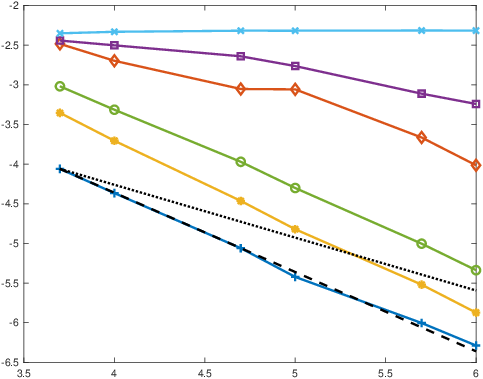}
         \caption{$X_4$}
     \end{subfigure}
     \hspace{0.15\textwidth}
     \begin{subfigure}[b]{0.46\textwidth}
         \centering
         \includegraphics[width=\textwidth]{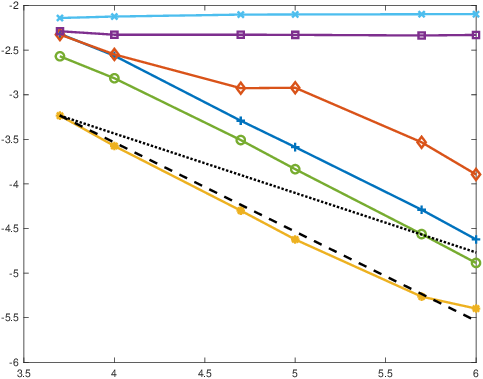}
         \caption{$X_5$}
     \end{subfigure}
\caption{Log--log MSE under LHS (y-axis) versus total budget $T$ (x-axis) for the hymod example. Reference slopes $-1$ and $-2/3$ indicate rates $\cO(T^{-1})$ and $\cO(T^{-2/3})$.}
     \label{fig:hymod_mse_lhs}
\end{figure}

Table~\ref{tab:lhs_performance} concludes this subsection by summarizing the performance of several estimators of the first-order Sobol' index numerator under LHS. When the Sobol' index is expected to be small, the L--NS and L--CR estimators are the most suitable choices. Conversely, when the index is anticipated to be large---specifically, greater than $0.4$---the L--NS estimator is generally preferred, with L--PF serving as a viable alternative. When the Sobol' index magnitude is unknown, we recommend the L--NS estimator, as it delivers robust performance and often outperforms L--PF and L--CR across a wide range of settings.

\begin{table}[t]
\centering
\caption{Performance summary and recommendations for estimators of the Sobol' index numerator under LHS.}
\label{tab:lhs_performance}
\begin{tabular}{ccccp{6.0cm}}
\hline
Estimator &  Biased for $V$  &  MSE  & Applicability & Recommendations \\ \hline

L--PF
& Yes
& $\mathcal{O}(T^{-1})$
& large indices ($>0.4$)
& Effective for large Sobol' indices \\
L--CR
& No
& $\mathcal{O}(T^{-1})$
& small indices ($<0.1$)
& Best choice for small indices \\

L--NS
& Yes
& $o(T^{-2/3})$
& all ranges
& Most robust and reliable across index magnitudes \\

L--JK
& Yes
& $\mathcal{O}(T^{-2/3})$
& not recommended
& Biased under LHS; exhibits inferior performance \\

L--SJ
& Yes
& $\mathcal{O}(1)$
& not recommended
& Bias dominates the MSE due to fixed inner-level sample size \\

L--OH
& Yes
& $\mathcal{O}(1)$
& not recommended
& Similar to L--SJ; bias prevents convergence \\ \hline

\end{tabular}
\end{table}

\section{Conclusions}
\label{sec:conclusion}

This paper studied Sobol' index estimation through a unified nested simulation framework under a fixed computational budget.
This perspective connects classical pick--freeze estimators with nested simulation estimators and shows how estimator structure, bias correction, budget allocation, and sampling design jointly determine performance. Under CMC with a budget of $T$ model evaluations, the standard nested simulation (NS) estimator attains the characteristic MSE rate $\mathcal{O}(T^{-2/3})$. The proposed unbiased jackknife (JK) estimator removes the leading bias but does not improve this rate, whereas the split jackknife (SJ) estimator attains the canonical rate $\mathcal{O}(T^{-1})$. This highlights the value of sample splitting when the centering term in the Sobol' index numerator must be estimated from simulation outputs.

We also characterized the impact of LHS on these estimators. For first-order Sobol' indices, LHS can improve the bias and variance of the NS estimator, yielding faster MSE convergence than under CMC. In contrast, LHS can be detrimental to bias-reduced NS estimators, including JK, SJ, and OH, because the dependence induced across outer-level scenarios can prevent their biases from vanishing unless the inner-level sample size $N$ grows with the total budget $T$. This condition is not needed under CMC. The numerical experiments corroborate these theoretical findings and illustrate their practical implications.

In practice, pick--freeze estimators can be effective when the index magnitude is known in advance: PF is preferable for large indices, whereas CR is preferable for small indices. When the index scale is unknown, however, these estimators can be less robust. Under CMC, JK and SJ provide stable performance across a wide range of Sobol' index values. Under LHS, their advantages diminish, and the standard NS estimator emerges as the most reliable choice among those studied.

Several directions remain for future research. Extending the LHS analysis and empirical comparisons to higher-order and total-effect indices would clarify the scope of the improvements observed for NS. Developing data-driven budget allocation rules under LHS could also help determine when $N$ should grow with $T$, particularly for bias-reduced estimators. Another direction is to study how other sampling designs, such as quasi-Monte Carlo (QMC) and randomized QMC, interact with nested simulation and bias reduction. Although these designs can improve integration accuracy in related sensitivity-analysis settings, their dependence structures may affect bias cancellation differently from both CMC and LHS. Finally, new nested simulation sampling designs that retain variance-reduction benefits while controlling dependence may yield further efficiency gains for large-scale sensitivity analysis.

{\footnotesize
\bibliographystyle{plainnat}
\bibliography{reference}

@article{sun2011efficient,
author = {Sun, Yunpeng and Apley, Daniel W. and Staum, Jeremy},
title = {Efficient Nested Simulation for Estimating the Variance of a Conditional Expectation},
journal = {Operations Research},
volume = {59},
number = {4},
pages = {998-1007},
year = {2011},
}

@article{gordy2010nested,
author = {Gordy, Michael B. and Juneja, Sandeep},
title = {Nested Simulation in Portfolio Risk Measurement},
journal = {Management Science},
volume = {56},
number = {10},
pages = {1833-1848},
year = {2010}
}

@article{liang2024fast,
author = {Liang, Guo and Zhang, Kun and Luo, Jun},
title = {A FAST Method for Nested Estimation},
journal = {INFORMS Journal on Computing},
volume = {36},
number = {6},
pages = {1481-1500},
year = {2024},
}

@article{sobol1990sensitivity,
  title={On Sensitivity Estimation for Nonlinear Mathematical Models},
  author={Sobol', Il'ya Meerovich},
  journal={Matematicheskoe Modelirovanie},
  volume={2},
  number={1},
  pages={112--118},
  year={1990},
  publisher={Russian Academy of Sciences, Branch of Mathematical Sciences}
}

@article{tarantola2007estimating,
  title={Estimating the Approximation Error When Fixing Unessential Factors in Global Sensitivity Analysis},
  author={Tarantola, Stefano and Gatelli, D and Kucherenko, SS and Mauntz, Wolfgang and others},
  journal={Reliability Engineering \& System Safety},
  volume={92},
  number={7},
  pages={957--960},
  year={2007},
  publisher={Elsevier}
}

@article{saltelli2010variance,
  title={Variance Based Sensitivity Analysis of Model Output. {D}esign and Estimator for the Total Sensitivity Index},
  author={Saltelli, Andrea and Annoni, Paola and Azzini, Ivano and Campolongo, Francesca and Ratto, Marco and Tarantola, Stefano},
  journal={Computer Physics Communications},
  volume={181},
  number={2},
  pages={259--270},
  year={2010},
  publisher={Elsevier}
}

@article{janon2014asymptotic,
  title={Asymptotic Normality and Efficiency of Two {S}obol' Index Estimators},
  author={Janon, Alexandre and Klein, Thierry and Lagnoux, Agnes and Nodet, Ma{\"e}lle and Prieur, Cl{\'e}mentine},
  journal={ESAIM: Probability and Statistics},
  volume={18},
  pages={342--364},
  year={2014},
  publisher={EDP Sciences}
}

@article{owen2013better,
author = {Owen, Art B.},
title = {Better estimation of small {S}obol' sensitivity indices},
year = {2013},
volume = {23},
number = {2},
journal = {ACM Transactions on Modeling and Computer Simulation},
pages={1--17},
}

@INPROCEEDINGS{ishigami1990importance,  
    author={Ishigami, Tsutoomu and Homma, Toshimitsu},  
    booktitle={the Proceedings of the First International Symposium on Uncertainty Modeling and Analysis},   
    title={An Importance Quantification Technique in Uncertainty Analysis for Computer Models},   
    year={1990},  
    volume={},  
    number={},  
    pages={398--403} 
}

@article{gamboa2016statistical,
  title={Statistical inference for {S}obol' pick-freeze {M}onte {C}arlo method},
  author={Gamboa, Fabrice and Janon, Alexandre and Klein, Thierry and Lagnoux, Agn{\`e}s and Prieur, Cl{\'e}mentine},
  journal={Statistics},
  volume={50},
  number={4},
  pages={881--902},
  year={2016},
  publisher={Taylor \& Francis}
}

@article{wagener2001framework,
  title={A framework for development and application of hydrological models},
  author={Wagener, Thorsten and Boyle, Douglas P and Lees, Matthew J and Wheater, Howard S and Gupta, Hoshin V and Sorooshian, Soroosh},
  journal={Hydrology and Earth System Sciences},
  volume={5},
  number={1},
  pages={13--26},
  year={2001},
  publisher={Copernicus GmbH}
}

@article{stein1987large,
 author = {Michael Stein},
 journal = {Technometrics},
 number = {2},
 pages = {143--151},
 publisher = {[Taylor & Francis, Ltd., American Statistical Association, American Society for Quality]},
 title = {Large Sample Properties of Simulations Using {L}atin hypercube sampling},
 volume = {29},
 year = {1987}
}

@article{goda2017computing,
title = {Computing the variance of a conditional expectation via non-nested {M}onte {C}arlo},
journal = {Operations Research Letters},
volume = {45},
number = {1},
pages = {63-67},
year = {2017},
author = {Takashi Goda}
}

@article{zhang2022bootstrap,
author = {Zhang, Kun and Liu, Guangwu and Wang, Shiyu},
title = {Technical Note—{B}ootstrap-based Budget Allocation for Nested Simulation},
journal = {Operations Research},
volume = {70},
number = {2},
pages = {1128-1142},
year = {2022},
}

@article{pianosi2015matlab,
title = {A {MATLAB} toolbox for Global Sensitivity Analysis},
journal = {Environmental Modelling \& Software},
volume = {70},
pages = {80-85},
year = {2015},
author = {Francesca Pianosi and Fanny Sarrazin and Thorsten Wagener}
}

@book{owen2013mcbook,
   author = {Art B. Owen},
   year = 2018,
   title = {Monte Carlo theory, methods and examples},
   publisher = {\url{https://artowen.su.domains/mc/}}
}

@article{damblin2021adaptive,
title = {Adaptive use of replicated {L}atin Hypercube Designs for computing {S}obol' sensitivity indices},
journal = {Reliability Engineering \& System Safety},
volume = {212},
pages = {107507},
year = {2021},
author = {Guillaume Damblin and Alberto Ghione}
}

@article{tissot2012estimating,
  title={Estimating {S}obol' indices combining {M}onte {C}arlo estimators and {L}atin hypercube sampling},
  author={Tissot, Jean-Yves and Prieur, Cl\'ementine},
  journal={Preprint available at hal},
  volume={743964},
  year={2012}
}

@article{gilquin2016recursive,
  title={Iterative estimation of {S}obol' indices based on replicated designs},
  author={Gilquin, Laurent and Prieur, Cl{\'e}mentine and Arnaud, Elise and Monod, Herv{\'e}},
  journal={Computational and Applied Mathematics},
  volume={40},
  number={1},
  pages={18},
  year={2021},
  publisher={Springer}
}

@article{gilquin2019making,
title = {Making the best use of permutations to compute sensitivity indices with replicated orthogonal arrays},
journal = {Reliability Engineering \& System Safety},
volume = {187},
pages = {28-39},
year = {2019},
author = {Laurent Gilquin and Elise Arnaud and Clémentine Prieur and Alexandre Janon}
}

@article{ehre2020framework,
title = {A framework for global reliability sensitivity analysis in the presence of multi-uncertainty},
journal = {Reliability Engineering \& System Safety},
volume = {195},
pages = {106726},
year = {2020},
author = {Max Ehre and Iason Papaioannou and Daniel Straub}
}

@article{puy2022comprehensive,
  title={A comprehensive comparison of total-order estimators for global sensitivity analysis},
  author={Puy, Arnald and Becker, William and Piano, Samuele Lo and Saltelli, Andrea},
  journal={International Journal for Uncertainty Quantification},
  volume={12},
  number={2},
  year={2022},
  pages = {1--18},
  publisher={Begel House Inc.}
}

@article{kouye2022exploiting,
  title={Exploiting deterministic algorithms to perform global sensitivity analysis of continuous-time {M}arkov chain compartmental models with application to epidemiology},
  author={Kouye, Henri Mermoz and Mazo, Gildas and Prieur, Cl{\'e}mentine and Vergu, Elisabeta},
  journal={arXiv preprint arXiv:2202.07277},
  year={2022}
}

@ARTICLE{liu2020identifying,
  author={Liu, Zhanlin and Banerjee, Ashis G. and Choe, Youngjun},
  journal={IEEE Transactions on Automation Science and Engineering}, 
  title={Identifying the Influential Inputs for Network Output Variance Using Sparse Polynomial Chaos Expansion}, 
  year={2021},
  volume={18},
  number={3},
  pages={1026-1036}
 }

@article{jaxa2021variance,
  title={Variance-based global sensitivity analysis and beyond in life cycle assessment: an application to geothermal heating networks},
  author={Jaxa-Rozen, Marc and Pratiwi, Astu Sam and Trutnevyte, Evelina},
  journal={The International Journal of Life Cycle Assessment},
  volume={26},
  number={5},
  pages={1008--1026},
  year={2021},
  publisher={Springer}
}

@article{giles2019multilevel,
  title={Multilevel nested simulation for efficient risk estimation},
  author={Giles, Michael B. and Haji-Ali, Abdul-Lateef},
  journal={SIAM/ASA Journal on Uncertainty Quantification},
  volume={7},
  number={2},
  pages={497--525},
  year={2019},
  publisher={SIAM}
}

@article{mckay1979comparison,
	author  = {McKay, Michael D. and Beckman, Richard J. and Conover, William J.},
	title   = {A Comparison of Three Methods for Selecting Values of Input Variables in the Analysis of Output from a Computer Code},
	journal = {Technometrics},
	volume  = {21},
	number  = {2},
	pages   = {239--245},
	year    = {1979}
}
}

\newpage

{\footnotesize
\appendix
\renewcommand{\theequation}{A.\arabic{equation}}
\setcounter{equation}{0}
\clearpage
\setcounter{page}{1}
\renewcommand{\thepage}{\arabic{page}}
\section{Proofs in Section \ref{sec:nest}}
\subsection{Proof of Proposition \ref{prop:bias_variance} in Subsection \ref{subsec:standard_ns}}\label{app:bias_variance_proof}
\begin{proof}
	For the bias of the NS estimator, we have
	\begin{equation}\label{eq:ns_bias_rewrite}
		 \E{V_{\mbox{\tiny NS}}^{\cu}} - V = \Var{L_N(\bX_\cu)} - V = V + \frac{\Var{\cY} - V}{N} - V=\frac{\Var{\cY} - V}{N} \ .
	\end{equation}
	Define $\mu \coloneqq \E{\cY}$. The variance of the NS estimator follows as
	\begin{equation}\label{eq:ns_variance_analysis}
		\Var{V_{\mbox{\tiny NS}}^{\cu}} = \frac{\E{L_N(\bX_\cu) - \mu}^4}{K} - \frac{(K-3)\Varp{L_{N}(\bX_\cu)}{2}}{K(K-1)} = \frac{E_4 - V^2}{K} + o\!\left(K^{-1}\right),
	\end{equation}
	where $E_4 \coloneqq \E{(\cY_{11}-\mu)(\cY_{12}-\mu)(\cY_{13}-\mu)(\cY_{14}-\mu)}$, and $\cY_{11}$, $\cY_{12}$, $\cY_{13}$ and $\cY_{14}$ represent independent simulation outputs generated conditional on the outer-level scenario $\bX_{\cu,1}$. Furthermore, we have
	\begin{align}
		E_4 =& \ \Esquare{\E{(\cY_{11} - \mu)(\cY_{12} - \mu)(\cY_{13} - \mu)(\cY_{14} - \mu) \mid \bX_\cu}} \nonumber \\
	=& \ \Esquare{\Ep{\cY - \mu \mid \bX_\cu}{4}} \\
	=&  \ \Epsquare{(\E{\cY \mid \bX_\cu} - \mu)^2}{2} + \operatorname{Var}\left[(\E{\cY \mid \bX_\cu} - \mu)^2\right]  \nonumber \\
	=& \ V^2 + \operatorname{Var}\left[(\E{\cY \mid \bX_\cu} - \mu)^2\right] \nonumber \\
	=&  \ V^2 + (\kappa_{\E{\cY \mid \bX_\cu}} - 1)V^2 = \kappa_{\E{\cY \mid \bX_\cu}} \cdot V^2, \label{eq:E4_decomposition}
	\end{align}
where the first equality on the right-hand side (RHS) follows from the definition of kurtosis. The proof is completed by combining \eqref{eq:ns_variance_analysis} and \eqref{eq:E4_decomposition}.
\end{proof}

\subsection{Proof of Proposition \ref{prop:jack} in Subsection \ref{subsec:unbias_jack}}\label{app:jack_proof}
\begin{proof}
 	We first analyze the bias of the JK estimator:
\begin{align}
	\E{V_{\mbox{\tiny JK}}^{\cu}} - V &= I \cdot \E{V_{\mbox{\tiny NS}}^{\cu}} - (I-1) \cdot \E{V_{\mbox{\tiny NS}, -1}^{\cu}} - V \nonumber \\
	&= I \cdot \Var{L_N(\bX_\cu)} - (I-1) \cdot\Var{L_{N,-1}(\bX_\cu)} - V  \label{eq:apply_ns_bias} \\
	&= I \cdot \left(V + \frac{\E{\Var{\cY \mid \bX_\cu}}}{N}\right) - (I-1) \cdot \left(V + \frac{\E{\Var{\cY \mid \bX_\cu}}}{N-N/I} \right)  - V \nonumber \\
	&= 0  \nonumber \ ,
\end{align}
where $L_{N,-1}(\bX_\cu)$ denotes an estimator of $L(\bX_\cu)$ in the same form as $L_N(\bX_\cu)$, constructed by omitting the outputs in the first section.

Define $b^{(l)} \coloneqq V_{\mbox{\tiny NS}}^{\cu} - V_{\mbox{\tiny NS},-l}^{\cu}$ for each $l \in [I]$. For the variance of the JK estimator, we have
\begin{align}
	\Var{V_{\mbox{\tiny JK}}^{\cu}} =& \ \Var{I V_{\mbox{\tiny NS}}^{\cu} - \frac{I-1}{I}\sum_{l=1}^{I}V_{\mbox{\tiny NS},-l}^{\cu}} = \Var{V_{\mbox{\tiny NS}}^{\cu} + \frac{I-1}{I}\sum_{l=1}^{I}b^{(l)}} \nonumber \\
	=& \ \Var{V_{\mbox{\tiny NS}}^{\cu}} + \frac{(I-1)^2}{I}\Var{b^{(1)}} \label{eq:jkvar_decompose_var} \\
	 & \ + 2(I-1) \Cov{V_{\mbox{\tiny NS}}^{\cu}}{b^{(1)}} + \frac{(I-1)^3}{I}\Cov{b^{(1)}}{b^{(2)}} \ . \label{eq:jkvar_decompose_cov}
\end{align}

For the term $\Var{V_{\mbox{\tiny NS}}^{\cu}}$ in \eqref{eq:jkvar_decompose_var}, we have
\begin{align}
	\Var{V_{\mbox{\tiny NS}}^{\cu}} =& \frac{E_4 + \mathcal{O}(N^{-1})}{K} - \frac{K-3}{K(K-1)} \left( \frac{\Varp{\cY}{2}}{N^2} + \frac{(N-1)^2}{N^2}\Varp{\E{\cY \mid \bX_\cu}}{2} \right. \nonumber \\
	 &\left. + \frac{2(N-1)}{N^2}\Var{\cY}\Var{\E{\cY \mid \bX_\cu}} \right) \nonumber \\
	 =& \frac{E_4}{K} + o(K^{-1}) \ ,\label{eq:jack_proof_e1}
\end{align}
where we recall that $E_4 = \E{(\cY_{11}-\mu)(\cY_{12}-\mu)(\cY_{13}-\mu)(\cY_{14}-\mu)}$ and $\mu = \E{\cY}$.

The term $\Var{b^{(l)}}$ in \eqref{eq:jkvar_decompose_var} can be rewritten as follows:
\begin{equation}\label{eq:variance_bl_rewritten}
	\Var{b^{(l)}} = \Var{V_{\mbox{\tiny NS}}^{\cu} - V_{\mbox{\tiny NS},-1}^{\cu}} = \E{\left(V_{\mbox{\tiny NS}}^{\cu} - V_{\mbox{\tiny NS},-1}^{\cu}\right)^2} - \mathbb{E}^2\left( V_{\mbox{\tiny NS}}^{\cu} - V_{\mbox{\tiny NS},-1}^{\cu} \right) \ .
\end{equation}
Since
\begin{equation}\label{eq:jk_ns_diff_rate}
	\E{V_{\mbox{\tiny NS}}^{\cu} - V_{\mbox{\tiny NS},-1}^{\cu}} = \Var{L_N(\bX_\cu)} - \Var{L_{N,-1}(\bX_\cu)} = - \frac{1}{N(I-1)} \E{\Var{\cY \mid \bX_\cu}} \ ,
\end{equation}
it follows that $\Ep{V_{\mbox{\tiny NS}}^{\cu} - V_{\mbox{\tiny NS},-1}^{\cu}}{2} = N^{-2} \cdot(I-1)^{-2}  \cdot \Ep{\Var{\cY \mid \bX_\cu}}{2}$.

For $\E{\left(V_{\mbox{\tiny NS}}^{\cu} - V_{\mbox{\tiny NS},-1}^{\cu} \right)^2}$, we have
\begin{align}
	\E{\left(V_{\mbox{\tiny NS}}^{\cu} - V_{\mbox{\tiny NS},-1}^{\cu} \right)^2} \leq \ & 2 \E{V_{\mbox{\tiny NS}}^{\cu} - \Var{L_N(\bX_\cu)} }^2  + 2 \E{V_{\mbox{\tiny NS},-1}^{\cu} - \Var{L_{N,-1}(\bX_\cu)}}^2 \nonumber \\
	&+ 2\left[\Var{L_N(\bX_\cu)} - \Var{L_{N,-1}(\bX_\cu)} \right]^2. \label{eq:esquare}
\end{align}
We next analyze the three terms on the RHS of \eqref{eq:esquare} in sequence. First, $\E{ V_{\mbox{\tiny NS}}^{\cu} - \Var{L_N(\bX_\cu)} }^2 = \Var{V_{\mbox{\tiny NS}}^{\cu}} = \cO(K^{-1})$. Similarly, $\E{ V_{\mbox{\tiny NS},-1}^{\cu} - \Var{L_{N,-1}(\bX_\cu)} }^2 =  \cO(K^{-1})$. Moreover, by \eqref{eq:ns_bias_rewrite}, $\big(\Var{L_N(\bX_\cu)} -\Var{L_{N,-1}(\bX_\cu)} \big)^2 = \cO(N^{-2}) $. Therefore, it follows from \eqref{eq:esquare} that
\begin{equation}\label{eq:jack_ns_different_second_moment_rate}
	\E{\left(V_{\mbox{\tiny NS}}^{\cu} - V_{\mbox{\tiny NS},-1}^{\cu} \right)^2} = \cO(K^{-1}) + \cO(N^{-2}).
\end{equation}
Furthermore, combining \eqref{eq:variance_bl_rewritten}, \eqref{eq:jk_ns_diff_rate}, and \eqref{eq:jack_ns_different_second_moment_rate} yields
\begin{equation}\label{eq:jack_proof_e2}
	\Var{b^{(l)}} = \frac{a^{\prime}}{K} + \frac{b^{\prime}}{N^2} + o(K^{-1}) + o(N^{-2}), \quad \forall l \in [I]
\end{equation}
for some positive constants $a^{\prime}$ and $b^{\prime}$.

For $\Cov{V_{\mbox{\tiny NS}}^{\cu}}{b^{(1)}}$ in \eqref{eq:jkvar_decompose_cov}, it follows from the Cauchy–Schwarz inequality that
\begin{equation}\label{eq:jack_proof_e3}
	\Cov{V_{\mbox{\tiny NS}}^{\cu}}{b^{(1)}} \leq \sqrt{\Var{V_{\mbox{\tiny NS}}^{\cu}}\Var{b^{(1)}}} = \cO(\max \{K^{-1}, K^{-1/2}N^{-1}\}).
\end{equation}
Similarly,
\begin{equation}\label{eq:jack_proof_e4}
	\Cov{b^{(1)}}{b^{(2)}} \leq \sqrt{\Var{b^{(1)}}\Var{b^{(2)}}} = \cO(\max \{K^{-1}, N^{-2}\}).
\end{equation}
Combining \eqref{eq:jack_proof_e1} through \eqref{eq:jack_proof_e4} yields
$
	\Var{V_{\mbox{\tiny JK}}^{\cu}} = a{K}^{-1} + bN^{-2} + o(K^{-1}) +o(N^{-2})
$
for some positive constants $a$ and $b$.
 \end{proof}

\subsection{Proof of Proposition \ref{prop:new_jack} in Subsection \ref{subsec:split_jack}}\label{app:new_jack_proof}
\begin{proof}
Notice that, with the preliminary dataset $\cD_{\mbox{\tiny pre}}$ used for estimating $\mu$, the bias and variance of $V_{\mbox{\tiny SJ}}^{\cu}$ can be written as follows:
\begin{align*}
 \E{V_{\mbox{\tiny SJ}}^{\cu}} - V = \E{\E{V_{\mbox{\tiny SJ}}^{\cu} \;\middle|\; \cD_{\mbox{\tiny pre}}}} - V \ ,  \quad
 \Var{V_{\mbox{\tiny SJ}}^{\cu}} = \Var{\E{V_{\mbox{\tiny SJ}}^{\cu} \;\middle|\; \cD_{\mbox{\tiny pre}}}} + \E{\Var{V_{\mbox{\tiny SJ}}^{\cu} \;\middle|\; \cD_{\mbox{\tiny pre}}}} \ .
\end{align*}
We first analyze $\E{V_{\mbox{\tiny SJ}}^{\cu} \;\middle|\; \cD_{\mbox{\tiny pre}}}$. Recall that $\widehat{\mu} $ denotes the sample mean obtained from $\cD_{\mbox{\tiny pre}}$. We have
\begin{equation}\label{eq:sj_conditional_mean}
	\E{V_{\mbox{\tiny SJ}}^{\cu} \;\middle|\; \cD_{\mbox{\tiny pre}}} = I \cdot \E{(L_N(\bX_\cu) - \widehat{\mu})^2 \;\middle|\; \cD_{\mbox{\tiny pre}}} - (I-1) \cdot \E{(L_{N,-1}(\bX_\cu) - \widehat{\mu})^2 \;\middle|\; \cD_{\mbox{\tiny pre}}} \ .
\end{equation}
On the one hand,
\begin{align}
	\E{(L_N(\bX_\cu) - \widehat{\mu})^2 \mid \cD_{\mbox{\tiny pre}}} =& \E{(L_N(\bX_\cu) - \mu + \mu - \widehat{\mu})^2 \;\middle|\; \cD_{\mbox{\tiny pre}}} \nonumber \\
	=&  \E{(L_N(\bX_\cu) - \mu)^2 \;\middle|\; \cD_{\mbox{\tiny pre}}} +  \E{(\mu - \widehat{\mu})^2 \;\middle|\; \cD_{\mbox{\tiny pre}}} \nonumber \\
	=&  \E{(L_N(\bX_\cu) - \mu)^2} + (\mu - \widehat{\mu})^2 \nonumber \\
	=& \Var{L_N(\bX_\cu)}  + (\mu - \widehat{\mu})^2 \ , \label{eq:sj_second_moment_decompose}
\end{align}
where the second and the last equalities on the RHS of \eqref{eq:sj_second_moment_decompose} follow from $\E{L_N(\bX_\cu)}=\mu$.

On the other hand,  $\E{(L_{N,-1}(\bX_\cu) - \widehat{\mu})^2 \;\middle|\; \cD_{\mbox{\tiny pre}}} = \Var{L_{N,-1}(\bX_\cu)}  + (\mu - \widehat{\mu})^2$.
It follows from \eqref{eq:sj_conditional_mean} and \eqref{eq:sj_second_moment_decompose} that $\E{V_{\mbox{\tiny SJ}}^{\cu} \;\middle|\; \cD_{\mbox{\tiny pre}}} = V + (\mu - \widehat{\mu})^2$, and hence
\begin{equation}\label{eq:sj_bias_final}
	\E{\E{V_{\mbox{\tiny SJ}}^{\cu} \mid \cD_{\mbox{\tiny pre}}} - V}  = \E{(\mu - \widehat{\mu})^2} = \frac{\Var{\cY}}{J} \ .
\end{equation}

For $\Var{V_{\mbox{\tiny SJ}}^{\cu}}$, it is easy to see that
\begin{equation}\label{eq:varsj_total}
	\Var{V_{\mbox{\tiny SJ}}^{\cu} \;\middle|\; \cD_{\mbox{\tiny pre}}} = \frac{1}{K} \Var{I \cdot (L_N(\bX_\cu) - \widehat{\mu})^2 - \frac{I-1}{I} \sum_{l=1}^{I} (L_{N,-l}(\bX_\cu)-\widehat{\mu})^2  \;\middle|\; \cD_{\mbox{\tiny pre}}}.
\end{equation}

Following steps analogous to those in the proof of Proposition \ref{prop:jack}, we have
\begin{align}
	& \Var{I \cdot (L_N(\bX_\cu) - \widehat{\mu})^2 - \frac{I-1}{I} \sum_{l=1}^{I} (L_{N, -l}(\bX_\cu)-\widehat{\mu})^2  \;\middle|\; \cD_{\mbox{\tiny pre}}} \nonumber \\
	=& \Var{(L_N(\bX_\cu) - \widehat{\mu})^2\;\middle|\; \cD_{\mbox{\tiny pre}}} + \frac{(I-1)^2}{I} \Var{b^{(1)} \;\middle|\; \cD_{\mbox{\tiny pre}}} \label{eq:sjvar_decompose_var}\\
	&+ 2(I-1) \Cov{(L_N(\bX_\cu) - \widehat{\mu})^2}{b^{(1)} \;\middle|\; \cD_{\mbox{\tiny pre}}} + \frac{(I-1)^3}{I} \Cov{b^{(1)}}{b^{(2)} \;\middle|\; \cD_{\mbox{\tiny pre}}}. \label{eq:sjvar_decompose_cov}
\end{align}
For the first term on the RHS of \eqref{eq:sjvar_decompose_var}, we have
\begin{align}
	\Var{(L_N(\bX_\cu) - \widehat{\mu})^2 \;\middle|\; \cD_{\mbox{\tiny pre}}}=& \ \Var{L_N^2(\bX_\cu) - 2\widehat{\mu}L_N(\bX_\cu)\;\middle|\; \cD_{\mbox{\tiny pre}}} \nonumber \\
	 \leq & \ 2\Var{L_N^2(\bX_\cu)} + 8 \widehat{\mu}^2 \Var{L_N(\bX_\cu)} \nonumber\\
	=& \ c_{1} + 8 \widehat{\mu}^2 \cdot \left(V + \frac{\E{\Var{\cY \;\middle|\; \bX_\cu}}}{N} \right) \ , \label{eq:sjvar_1}
\end{align}
where $c_{1}$ is some positive constant.
For the term $\Var{b^{(1)} \mid \cD_{\mbox{\tiny pre}}}$ on the RHS of \eqref{eq:sjvar_decompose_var}, we first write it as
\begin{align}\label{eq:new_jack_varb}
	\Var{b^{(1)} \;\middle|\; \cD_{\mbox{\tiny pre}}} =& \ \E{ \left(  (L_N(\bX_\cu) - \widehat{\mu})^2  - (L_{N,-1}(\bX_\cu)-\widehat{\mu})^2 \right)^2 \;\middle|\; \cD_{\mbox{\tiny pre}}} \ - \Ep{ (L_N(\bX_\cu) - \widehat{\mu})^2 - (L_{N, -1}(\bX_\cu)-\widehat{\mu})^2 \;\middle|\; \cD_{\mbox{\tiny pre}}}{2} \nonumber \\
	 \coloneqq  & \  \mbox{item }(i) - \mbox{item }(ii) \ .
\end{align}
Since
\begin{align}
	\mathbb{E} \bigg( (L_N(\bX_\cu) - \widehat{\mu})^2 - (L_{N, -1}(\bX_\cu)-\widehat{\mu})^2 \bigm| \cD_{\mbox{\tiny pre}} \bigg) =& \Var{L_N(\bX_\cu)} + (\mu - \widehat{\mu})^2 - \Var{L_{N, -1}(\bX_\cu)}  - (\mu - \widehat{\mu})^2 \nonumber \\
	=& - \frac{1}{N(I-1)}\E{\Var{\cY \bigm| \bX_\cu}}, \label{eq:sjvar_2}
\end{align}
it follows that $ \mbox{item } (ii) = N^{-2} \cdot (I-1)^{-2}\Ep{\Var{\cY \mid \bX_\cu}}{2}$. For item $(i)$, we have
\begin{align}
	& \E{ \left(  (L_N(\bX_\cu) - \widehat{\mu})^2 - (L_{N, -1}(\bX_\cu)-\widehat{\mu})^2 \right)^2 \;\middle|\; \cD_{\mbox{\tiny pre}}}  \nonumber\\
	\leq& \ 3 \E{L_N^4(\bX_\cu)} + 3 \E{(L_{N, -1}(\bX_\cu))^4} + 6\widehat{\mu}^2 \E{(L_N(\bX_\cu) - L_{N, -1}(\bX_\cu))^2} \nonumber \\
	\leq & \ c_{2} + c_{3} \cdot \widehat{\mu}^2 \label{eq:sjvar_3}  \ ,
\end{align}
where $c_2$ and $c_3$ are some positive constants; here, the first inequality follows from the Cauchy–Schwarz inequality, and the second one uses the boundedness of the fourth moment of the model output.

Combining \eqref{eq:varsj_total}, \eqref{eq:sjvar_1}, \eqref{eq:sjvar_2}, and \eqref{eq:sjvar_3} yields
\begin{equation}\label{eq:sj_variance_final}
	\Var{V_{\mbox{\tiny SJ}}^{\cu} \;\middle|\; \cD_{\mbox{\tiny pre}}} \leq \frac{c_4 + c_5 \widehat{\mu}^2}{K} \ ,
\end{equation}
where $c_4$ and $c_5$ are some positive constants. Combining \eqref{eq:sj_bias_final} and \eqref{eq:sj_variance_final} yields
\begin{align*}
	\Var{V_{\mbox{\tiny SJ}}^{\cu}}  = \E{\Var{V_{\mbox{\tiny SJ}}^{\cu} \;\middle|\; \cD_{\mbox{\tiny pre}}}} + \Var{\E{V_{\mbox{\tiny SJ}}^{\cu} \;\middle|\; \cD_{\mbox{\tiny pre}}}}
 = \frac{c}{K} +o(K^{-1}) + \frac{d}{J^2} + o(J^{-2}) \ ,
\end{align*}
where $c = \E{c_4 + c_5 \widehat{\mu}^2}$, and $d$ is some positive constant.
\end{proof}

\section{Proofs in Section \ref{sec:ns_lhs}}
We first present several auxiliary lemmas in Subsection \ref{app:aux_lemma} that will be useful for the subsequent proofs. The proofs for Section \ref{sec:ns_lhs} are then provided in the following subsections.
To facilitate analysis, we assume that
\(\lhs{\bX}_{\cu,i} \sim U[0,1)^{|\cu|}\) and \(\lhs{\bX}_{-\cu,i} \sim U[0,1)^{|-\cu|}\). For input vectors with non-uniformly distributed components, LHS can generate observations from the target distributions by applying the inverse cumulative distribution function transformation. As shown in the proof of Theorem 1 in \cite{stein1987large}, this adjustment does not affect the theoretical results.

\subsection{Auxiliary Lemmas}\label{app:aux_lemma}
Recall from \eqref{eq:lhs_ns_est} that $\lhs{L}_N(\lhs{\bX}_{\cu,i}) \coloneqq N^{-1} \sum_{j=1}^{N} \lhs{\cY}_{ij}$ for each $i \in [K]$, where $\lhs{\cY}_{ij} = f(\lhs{\bX}_{\cU, i}, \lhs{\bX}_{-\cU, j})$.

\begin{lemma}[Functional ANOVA \citep{owen2013mcbook}]\label{lem:functional_anova_template}
	Let $\phi(\bZ)$ be a square-integrable function of independent inputs $\bZ=(Z_1, Z_2,\ldots,Z_d)$ with dimensionality $d$. Then $\phi$ admits the functional-ANOVA decomposition
	\[
	\phi(\bZ) = \phi_0 + \sum_{j=1}^d \phi_j(Z_j) + \sum_{1 \leq j < k \leq d} \phi_{jk}(Z_j,Z_k) + \cdots \ ,
	\]
	where
$\phi_0 = \E{\phi(\bZ)}$,  $\phi_j(Z_j) = \E{\phi(\bZ) \mid Z_j} - \phi_0$,
	and
	$
	\phi_{jk}(Z_j,Z_k) = \E{\phi(\bZ) \mid Z_j, Z_k} - \phi_j(Z_j) - \phi_k(Z_k) - \phi_0.
	$
	If $\delta(\bZ) \coloneqq \phi(\bZ) - \phi_0 - \sum_{j=1}^d \phi_j(Z_j)$, then $\E{\delta(\bZ)}=0$ and
	$
	\Var{\delta(\bZ)} = \Var{\phi(\bZ)} - \sum_{j=1}^d \Var{\phi_j(Z_j)}.
	$
\end{lemma}

\begin{lemma}\label{lem:L_unbias}
	The estimator $\lhs{L}_N(\lhs{\bX}_{\cu,i})$ is unbiased for $L(\lhs{\bX}_{\cu,i})$ for each $i \in [K]$.
\end{lemma}

\begin{proof}
	Using Theorem 10.1 of \cite{owen2013mcbook}, we have $\lhs{\bX}_{\cu,i} \sim U[0, 1)^{|\cU|}$ for each $i \in [K]$. It follows that $\E{\lhs{L}_N(\lhs{\bX}_{\cu,i})} = \E{N^{-1}\sum_{j=1}^{N}\lhs{\cY}_{ij} \mid \lhs{\bX}_{\cu,i}} = \E{\lhs{\cY}_{i1} \mid \lhs{\bX}_{\cu,i}} = L(\lhs{\bX}_{\cu,i}), \ \forall i \in [K]$.
\end{proof}

\begin{lemma}\label{lem:V_bias}
	The expectation of $V_{\tiny \mbox{L--NS}}^{\cu}$ satisfies
$
	\E{V_{\tiny \mbox{L--NS}}^{\cu}} = \Var{\lhs{L}_N(\lhs{\bX}_\cu)} + \cO(K^{-1}).
$
	Moreover, if $|\cu| = 1$, the expectation refines to
	$
	\E{V_{\tiny \mbox{L--NS}}^{\cu}} = \Var{\lhs{L}_N(\lhs{\bX}_\cu)} + \cO(K^{-1}N^{-1}) + \cO(K^{-3}).
	$
\end{lemma}

\begin{proof}
	Since $V_{\tiny \mbox{L--NS}}^{\cu} = K^{-1}\sum_{i=1}^{K} \left(\lhs{L}_N(\lhs{\bX}_{\cu,i}) - K^{-1}\sum_{j=1}^{K}\lhs{L}_N(\lhs{\bX}_{\cu,j}) \right)^2$, we have
	\begin{align}
		\E{V_{\tiny \mbox{L--NS}}^{\cu}} &= \E{\lhs{L}_N(\lhs{\bX}_{\cu})}^2 - \E{ \left( K^{-1}\sum_{j=1}^{K}\lhs{L}_N(\lhs{\bX}_{\cu,j}) \right)^2} \nonumber \\
		&= \Var{\lhs{L}_N(\lhs{\bX}_{\cu})} + \Ep{\lhs{L}_N(\lhs{\bX}_{\cu})}{2} -  \Var{K^{-1}\sum_{j=1}^{K}\lhs{L}_N(\lhs{\bX}_{\cu,j})} -  \Ep{K^{-1}\sum_{j=1}^{K}\lhs{L}_N(\lhs{\bX}_{\cu,j})}{2} \nonumber \\
		&= \Var{\lhs{L}_N(\lhs{\bX}_{\cu})}  - \Var{K^{-1}\sum_{j=1}^{K}\lhs{L}_N(\lhs{\bX}_{\cu,j})}  \label{eq:expectation_nest_decomposition} \ .
	\end{align}
Regarding the second term on the RHS of \eqref{eq:expectation_nest_decomposition}, we have
	\begin{align}
	  \Var{K^{-1}\sum_{j=1}^{K}\lhs{L}_N(\lhs{\bX}_{\cu,i})}
		 = \ & \E{\Var{K^{-1}\sum_{j=1}^{K}\lhs{L}_N(\lhs{\bX}_{\cu,j})\middle| \{ \lhs{\bX}_{\cu,j} \}_{j \in [K]}  }}  + \Var{K^{-1}\sum_{j=1}^{K} \E{\lhs{L}_N(\lhs{\bX}_{\cu,j}) \middle|  \{ \lhs{\bX}_{\cu,j} \}_{j \in [K]} }} \nonumber  \\
		 \coloneqq \ & \mbox{item }(i) + \mbox{item }(ii) \label{eq:total_variance_LN} \ .
	\end{align}
	For item $(i)$ in \eqref{eq:total_variance_LN}, we have
$
	\E{\Var{K^{-1}\sum_{j=1}^{K}\lhs{L}_N(\lhs{\bX}_{\cu,j})\middle| \{ \lhs{\bX}_{\cu,j} \}_{j \in [K]}  }} = K^{-1}\E{\Var{\lhs{L}_N(\lhs{\bX}_{\cu}) \middle|  \lhs{\bX}_{\cu}}}  = \cO(T^{-1})$,
	where the first equality follows from the conditional independence of $\lhs{L}_N(\lhs{\bX}_{\cu,j})$ across $j \in [K]$ given $\lhs{\bX}_{\cu,j}$, and the last equality uses $NK = T$.

	For item $(ii)$ in \eqref{eq:total_variance_LN}, it follows from Corollary 1 of \cite{stein1987large} that
		\begin{align*}
			\Var{K^{-1}\sum_{j=1}^{K} \E{\lhs{L}_N(\lhs{\bX}_{\cu,j}) \middle|  \{ \lhs{\bX}_{\cu,j} \}_{j \in [K]} }} = K^{-1} \int r^2(\lhs{\bX}_\cu)d\lhs{\bX}_{\cu} + o(K^{-1}) \ ,
		\end{align*}
	where, by applying Lemma \ref{lem:functional_anova_template} to $\phi(\lhs{\bX}_\cu)=\E{\lhs{\cY}\mid \lhs{\bX}_\cu}$, $r(\lhs{\bX}_{\cu})$ is the remainder after removing the mean and all main effects:
$
	r(\lhs{\bX}_{\cu}) = \E{\lhs{\cY} \bigm| \lhs{\bX}_\cu} - \sum_{i \in \cu}\Es{\lhs{\cY} \bigm| \lhs{\bX}_{i}}{\lhs{\bX}_{i}} - \E{\lhs{\cY}}
$.
	Specifically, when $\cu$ is a singleton, $r(\lhs{\bX}_{\cu}) = 0$. In this case, LHS reduces to stratified sampling in one dimension, which achieves a variance convergence rate of $\cO(K^{-3})$ \citep{owen2013mcbook}. Therefore, we have
$
	\Var{K^{-1}\sum_{i=1}^{K} \lhs{L}_N(\lhs{\bX}_{\cu,i})} = \cO(T^{-1}) + \cO(K^{-3})$    when  $|\cu| = 1$.
\end{proof}

\begin{lemma}\label{lem:L_i_unbias}
	For fixed $i \in [K]$, the estimator $\lhs{L}_{N,-l}(\lhs{\bX}_{\cu,i})$ is unbiased for $L(\lhs{\bX}_{\cu,i})$ for each $l \in [I]$, where $I$ denotes the number of sections used in the L-JK and L-SJ estimators.
\end{lemma}
\begin{proof}
By Theorem 10.1 of \cite{owen2013mcbook}, $\bX_{-\cu,i}$ is uniformly distributed on $[0,1)^{|-\cu|}$ for $i \in [N]$. For an arbitrary nonempty subset $\cS \subseteq [N]$ with $|\cS|=N-(N/I)$, $\bX_{-\cu, i} \sim U[0,1)^{|-\cu|}$ also holds for each $i \in \cS$. Hence,  $\E{\lhs{L}_{N,-l}(\lhs{\bX}_{\cu,i})} = |\cS|^{-1} \sum_{j \in \cS} \E{\lhs{\cY}_{ij} \mid \lhs{\bX}_{\cu,i}} = \E{\lhs{\cY}_{i1} \mid \lhs{\bX}_{\cu,i}} = L(\lhs{\bX}_{\cu,i})$ for all $i \in [K]$.
\end{proof}

\begin{lemma}\label{lem:additive_approx}
	Given $\lhs{\bX}_{\cu,i} = \bx_{\cu,i}$, for the estimator $\lhs{L}_{N,-l}(\bx_{\cu,i})$ with $l \in [I]$, we have
	\[\Var{\lhs{L}_{N,-l}(\bx_{\cu,i})} = \frac{1}{N-(N/I)} \int e^2(\bx_{\cu,i}, \lhs{\bX}_{-\cu})d \lhs{\bX}_{-\cu} + r_N + \cO(N^{-2})\ ,\]
	where, upon applying Lemma \ref{lem:functional_anova_template} to the function $\phi(\bX_{-\cu}) \coloneqq f(\bx_{\cu,i}, \bX_{-\cu})$,
	$
	e(\bx_{\cu,i}, \lhs{\bX}_{-\cu}) = f(\bx_{\cu,i}, \lhs{\bX}_{-\cu}) - \sum_{j \in -\cu} f_j(\bx_{\cu,i}, \lhs{\bX}_{j}) - \E{f(\bx_{\cu,i}, \lhs{\bX}_{-\cu})}
$
	is the remainder after removing the mean and all main effects. Moreover, $r_N = o(N^{-1})$, and $r_N > 0$ for finite $N$.
\end{lemma}
\begin{proof}
	The proof proceeds analogously to that of Corollary 1 in \cite{stein1987large}.
\end{proof}

\subsection{Proofs in Subsection \ref{subsec:lhs-pfcr}}
\subsubsection{Proof of Proposition \ref{prop:pf_lhs_bias_variance}} \label{app:proof_pf_lhs_bias_variance}

\begin{proof}
We first analyze the bias of $V_{\mbox{\tiny L--PF}}^{\cu}$. Define
$
\overline{Z} \coloneqq {K}^{-1} \sum_{i=1}^K Z_i$, with   $Z_i \coloneqq \left(f(\lhs{\bX}_{\cu,i}, \lhs{\bX}_{-\cu,i})+f(\lhs{\bX}_{\cu,i} , \lhs{\bX}_{-\cu,i}^{\prime})\right)$/2,  for   $i \in [K]$. It follows that
$
V_{\mbox{\tiny L--PF}}^{\cu} = {K}^{-1} \sum_{i=1}^K  f(\lhs{\bX}_{\cu,i}, \lhs{\bX}_{-\cu,i}) f(\lhs{\bX}_{\cu,i}, \lhs{\bX}_{-\cu,i}^{\prime}) - \overline{Z}^{\,2}$.
On the one hand,
\begin{align*}
	\E{\frac{1}{K} \sum_{i=1}^K f(\lhs{\bX}_{\cu,i}, \lhs{\bX}_{-\cu,i}) f(\lhs{\bX}_{\cu,i}, \lhs{\bX}_{-\cu,i}^{\prime})}
& = \E{f(\lhs{\bX}_{\cu,i}, \lhs{\bX}_{-\cu,i}) f(\lhs{\bX}_{\cu,i}, \lhs{\bX}_{-\cu,i}^{\prime})} \ .
\end{align*}
On the other hand,
 $\E{\overline{Z}} = \E{f(\lhs{\bX}_{\cu,i}, \lhs{\bX}_{-\cu,i})}$. It follows that
$
\E{\overline{Z}^{2}} = \Var{\overline{Z}} + \left(\E{\overline{Z}}\right)^2 = \Var{\overline{Z}} + \Ep{f(\lhs{\bX}_{\cu,i}, \lhs{\bX}_{-\cu,i})}{2}
$.
Hence,
$
\E{V_{\mbox{\tiny L--PF}}^{\cu}} = V - \Var{\overline{Z}}.
$
Since $\Var{\overline{Z}} = \cO(K^{-1})$ by Theorem 1 in \cite{stein1987large}, we have
$\E{V_{\mbox{\tiny L--PF}}^{\cu}} - V = \cO(K^{-1}).$

We next analyze the variance of $V_{\mbox{\tiny L--PF}}^{\cu}$. We write $V_{\mbox{\tiny L--PF}}^{\cu} = A_K - B_K$, where $A_K \coloneqq K^{-1} \sum_{i=1}^K f(\lhs{\bX}_{\cu,i}, \lhs{\bX}_{-\cu,i}) f(\lhs{\bX}_{\cu,i}, \lhs{\bX}_{-\cu,i}^{\prime})$, and $B_K \coloneqq \overline{Z}^{2}$. It follows that
\begin{equation}\label{eq:var_L_PF_ineq}
	\Var{V_{\mbox{\tiny L--PF}}^{\cu}} = \Var{A_K - B_K} \leq 2 \Var{A_K} + 2 \Var{B_K} \ .
\end{equation}
Define
$
h_1(\lhs{\bX}_{\cu}, \lhs{\bX}_{-\cu}, \lhs{\bX}_{-\cu}^{\prime}) := f(\lhs{\bX}_{\cu,i}, \lhs{\bX}_{-\cu,i}) f(\lhs{\bX}_{\cu,i}, \lhs{\bX}_{-\cu,i}^{\prime}).
$
It follows from Theorem 1 in \cite{stein1987large} that
\begin{equation}\label{eq:var_AK}
	\Var{A_K} = \Var{\frac{1}{K} \sum_{i=1}^K h_1(\lhs{\bX}_{\cu, i}, \lhs{\bX}_{-\cu, i}, \lhs{\bX}_{-\cu, i}^{\prime})} = \cO(K^{-1}) \ .
\end{equation}
Similarly, define
$
h_2(\lhs{\bX}_{\cu}, \lhs{\bX}_{-\cu}, \lhs{\bX}_{-\cu}^{\prime}) :=  \left( f(\lhs{\bX}_{\cu}, \lhs{\bX}_{-\cu}) + f(\lhs{\bX}_{\cu}, \lhs{\bX}_{-\cu}^{\prime}) \right)/2
$.
Then, we have
\begin{equation}\label{eq:var_Z_bar}
	\Var{\overline{Z}} = \Var{\frac{1}{K}\sum_{k=1}^{K} h_2(\lhs{\bX}_{\cu, k}, \lhs{\bX}_{-\cu, k}, \lhs{\bX}_{-\cu, k}^{\prime}) } = \cO(K^{-1}) \ .
\end{equation}
Notice that $B_K$ can be written as $g(\overline{Z})$ where $g(t) = t^2$. Since $g$ is smooth and $\overline{Z}$ has finite second moment, the Taylor series expansion yields
$
\Var{B_K}
= \Var{g(\overline{Z})}
= g^{\prime}(\E{\overline{Z}})^2 \, \Var{\overline{Z}} + \cO\big(\Var{\overline{Z}}^{2}\big)
$.
Therefore, it follows from \eqref{eq:var_Z_bar} that
\begin{equation}\label{eq:var_BK}
	\Var{B_K}
= \cO\left(\Var{\overline{Z}}\right) + \cO\left(\Var{\overline{Z}}^{\,2}\right)
= \cO\left(K^{-1}\right) \ .
\end{equation}
Combining \eqref{eq:var_AK} and \eqref{eq:var_BK} gives
$
\Var{V_{\mbox{\tiny L--PF}}^{\cu}} \leq 2 \Var{A_K} + 2 \Var{B_K} = \cO(K^{-1})
$.
\end{proof}

\subsubsection{Proof of Proposition \ref{prop:cr_lhs_bias_variance}} \label{app:proof_cr_lhs_bias_variance}

\begin{proof}
We first analyze the bias of $V_{\mbox{\tiny CR}}^{\cu}$. Notice that
\begin{align*}
	\E{\frac{1}{K} \sum_{i=1}^K f(\lhs{\bX}_{\cu,i}, \lhs{\bX}_{-\cu,i}) f(\lhs{\bX}_{\cu,i}, \lhs{\bX}_{-\cu,i}^{\prime})} = \E{f(\lhs{\bX}_{\cu,i}, \lhs{\bX}_{-\cu,i}) f(\lhs{\bX}_{\cu,i}, \lhs{\bX}_{-\cu,i}^{\prime})} \ .
\end{align*}
Similarly,
\begin{align*}
	& \E{\frac{1}{K} \sum_{i=1}^K f(\lhs{\bX}_{\cu,i}, \lhs{\bX}_{-\cu,i}) f(\lhs{\bX}_{\cu,i}^{\prime}, \lhs{\bX}_{-\cu,i}^{\prime})} = \E{f(\lhs{\bX}_{\cu,i}, \lhs{\bX}_{-\cu,i})} \E{ f(\lhs{\bX}_{\cu,i}^{\prime}, \lhs{\bX}_{-\cu,i}^{\prime})} \ , \\
	& \E{\frac{1}{K} \sum_{i=1}^K f(\lhs{\bX}_{\cu,i}^{\prime\prime}, \lhs{\bX}_{-\cu,i}) f(\lhs{\bX}_{\cu,i}, \lhs{\bX}_{-\cu,i}^{\prime})} = \E{f(\lhs{\bX}_{\cu,i}^{\prime\prime}, \lhs{\bX}_{-\cu,i})} \E{ f(\lhs{\bX}_{\cu,i}, \lhs{\bX}_{-\cu,i}^{\prime})} \ ,  \\
	& \E{\frac{1}{K} \sum_{i=1}^K f(\lhs{\bX}_{\cu,i}^{\prime\prime}, \lhs{\bX}_{-\cu,i}) f(\lhs{\bX}_{\cu,i}^{\prime}, \lhs{\bX}_{-\cu,i}^{\prime})} = \E{f(\lhs{\bX}_{\cu,i}^{\prime\prime}, \lhs{\bX}_{-\cu,i})} \E{ f(\lhs{\bX}_{\cu,i}^{\prime}, \lhs{\bX}_{-\cu,i}^{\prime})} \ .
\end{align*}
Hence, we have
\[\E{V_{\mbox{\tiny L--CR}}^{\cu}} - V = \E{f(\lhs{\bX}_{\cu,i}, \lhs{\bX}_{-\cu,i}) f(\lhs{\bX}_{\cu,i}, \lhs{\bX}_{-\cu,i}^{\prime})} - \Ep{ f(\lhs{\bX}_{\cu,i}, \lhs{\bX}_{-\cu,i})}{2} - V = 0 \ .\]
Now we analyze the variance of $V_{\mbox{\tiny CR}}^{\cu}$. Similar to the proof in Appendix \ref{app:proof_pf_lhs_bias_variance}, we can show that
\[\E{K^{-1} \sum_{i=1}^K f(\lhs{\bX}_{\cu,i}, \lhs{\bX}_{-\cu,i}) f(\lhs{\bX}_{\cu,i}, \lhs{\bX}_{-\cu,i}^{\prime})} = \cO(K^{-1}) \ .\]
The other terms can be proved analogously.
\end{proof}

\subsection{Proofs in Subsection \ref{subsec:standard_nest}}
\subsubsection{Proof of Proposition \ref{prop:lhs_nest_bias}}\label{app:proof_nest_lhs_bias}

\begin{proof}

By Corollary 1 of \cite{stein1987large}, conditional on $\lhs{\bX}_\cu$, we have
\begin{equation}\label{eq:var_LN_stein}
	\Vars{\lhs{L}_N(\lhs{\bX}_\cu) \bigm|  \lhs{\bX}_{\cu} }{\lhs{\bX}_{-\cu} } = \frac{1}{N} \int e^2(\lhs{\bX}_\cu, \lhs{\bX}_{-\cu})d\lhs{\bX}_{-\cu} + r_N + \cO(N^{-2}) \ ,
\end{equation}
where $r_N=o(N^{-1})$ and $r_N > 0$. To express $e(\lhs{\bX}_\cu, \lhs{\bX}_{-\cu})$ explicitly, we fix $x_\cu \in [0,1)^{|\cu|}$, consider $
\phi_{x_\cu}(\bX_{-\cu}) \coloneqq f(x_\cu,\bX_{-\cu}) ,
$
and apply Lemma \ref{lem:functional_anova_template} to $\phi_{x_\cu}$. The resulting mean and main effects are
\begin{align*}
    \phi_0(x_\cu) &= \E{\phi_{x_\cu}(\bX_{-\cu})} = \E{\cY \bigm| \bX_\cu = x_\cu}, \\
    \phi_j(x_\cu;X_j) &= \E{\phi_{x_\cu}(\bX_{-\cu}) \bigm| X_j} - \phi_0(x_\cu) \\
    &= \E{\cY \bigm| X_j, \bX_\cu=x_\cu} - \E{\cY \bigm| \bX_\cu=x_\cu}, \ \mbox{for } j \in -\cu \ .
\end{align*}

Let
$
\delta(x_\cu,\bX_{-\cu}) \coloneqq \phi_{x_\cu}(\bX_{-\cu}) - \phi_0(x_\cu) - \sum_{j \in -\cu} \phi_j(x_\cu;X_j)
$
denote the remainder after removing the mean and all main effects. We then evaluate this remainder at the randomly sampled outer-level scenarios via LHS and define
$
e(\lhs{\bX}_\cu, \lhs{\bX}_{-\cu}) \coloneqq \delta(\lhs{\bX}_\cu,\lhs{\bX}_{-\cu}).
$
Equivalently, we have
\[
e(\lhs{\bX}_\cu, \lhs{\bX}_{-\cu})
= f(\lhs{\bX}_\cu, \lhs{\bX}_{-\cu}) - \E{\cY \bigm| \bX_\cu=\lhs{\bX}_\cu}
- \sum_{j \in -\cu}\left(\E{\cY \bigm| \lhs{X}_j,\bX_\cu=\lhs{\bX}_\cu} - \E{\cY \bigm| \bX_\cu=\lhs{\bX}_\cu}\right) \ .
\]
It follows from Lemma~\ref{lem:functional_anova_template} that  $\Es{e(\lhs{\bX}_\cu, \lhs{\bX}_{-\cu})}{\lhs{\bX}_{-\cu}} = 0$ and
\begin{align}
	\int e^2(\lhs{\bX}_\cu, \lhs{\bX}_{-\cu})d\bX_{-\cu} &= \Var{e(\lhs{\bX}_\cu, \lhs{\bX}_{-\cu})} \nonumber \\
	&= \Var{\lhs{\cY} \bigm| \bX_\cu=\lhs{\bX}_\cu} -\sum_{j \in -\cu}\Var{\phi_j(\lhs{\bX}_\cu;\lhs{X}_j)} \nonumber \\
	&= \Var{\lhs{\cY} \bigm| \bX_\cu=\lhs{\bX}_\cu} -\sum_{j \in -\cu}\Var{\E{\lhs{\cY} \bigm| \lhs{X}_j, \bX_\cu=\lhs{\bX}_\cu}}. \label{eq:e_square_equality}
\end{align}
Define $R_{\cu} \coloneqq \Essquare{\Vars{e(\lhs{\bX}_\cu, \lhs{\bX}_{-\cu}) \mid \lhs{\bX}_\cu}{\lhs{\bX}_{-\cu} } }{\lhs{\bX}_\cu} / \Essquare{\Vars{\lhs{\cY} \mid \lhs{\bX}_{\cu}}{\lhs{\bX}_{-\cu}}}{\lhs{\bX}_\cu}$. It follows that
\begin{equation}
\begin{split}
R_{\cu} &= \frac{\Essquare{\Vars{\lhs{\cY} \bigm| \lhs{\bX}_{\cu}}{\lhs{\bX}_{-\cu}}}{\lhs{\bX}_\cu}
- \sum_{j \in -\cu}\Essquare{\Vars{\Es{\lhs{\cY} \bigm| \lhs{X}_j, \lhs{\bX}_\cu}{\lhs{\bX}_{-\cu \setminus \{j\}}}}{\lhs{X}_j}}{\lhs{\bX}_\cu}}
{\Essquare{\Vars{\lhs{\cY} \bigm| \lhs{\bX}_{\cu}}{\lhs{\bX}_{-\cu}}}{\lhs{\bX}_\cu}}  \\
&= 1 - \frac{\sum_{j \in -\cu}\Essquare{\Vars{\Es{\lhs{\cY} \bigm| \lhs{X}_j, \lhs{\bX}_\cu}{\lhs{\bX}_{-\cu \setminus \{j\}}}}{\lhs{X}_j}}{\lhs{\bX}_\cu}}
{\Essquare{\Vars{\lhs{\cY} \bigm| \lhs{\bX}_{\cu}}{\lhs{\bX}_{-\cu}}}{\lhs{\bX}_\cu}} \\
&= 1 - \frac{\sum_{j \in -\cu} \left(
\Essquare{\Vars{\lhs{\cY} \bigm| \lhs{\bX}_\cu}{\lhs{\bX}_{-\cu}}}{\lhs{\bX}_\cu}
-
\Essquare{\Es{ \Vars{\lhs{\cY} \bigm| \lhs{X}_j, \lhs{\bX}_\cu}{\lhs{\bX}_{-\cu \setminus \{j\}}} }{\lhs{X}_j}}{\lhs{\bX}_\cu}
\right)}
{\Essquare{\Vars{\lhs{\cY} \bigm| \lhs{\bX}_{\cu}}{\lhs{\bX}_{-\cu}}}{\lhs{\bX}_\cu}} \\
&= 1 - \frac{\sum_{j \in -\cu} \left(
\Varssquare{\Es{\lhs{\cY} \bigm| \lhs{\bX}_\cu, \lhs{X}_j}{\lhs{\bX}_{-\cu \setminus \{j\}}}}{\lhs{X}_{\cu \cup \{j\}}}
-
\Varssquare{\Es{\lhs{\cY} \bigm| \lhs{\bX}_\cu}{\lhs{\bX}_{-\cu}}}{\lhs{\bX}_\cu}
\right)}
{\Essquare{\Vars{\lhs{\cY} \bigm| \lhs{\bX}_{\cu}}{\lhs{\bX}_{-\cu}}}{\lhs{\bX}_\cu}} \\
&= \frac{S_T^{-\cu} - \sum_{j \in -\cu}(S^{j} + S^{\cu, j})}{S_T^{-\cu}} \ ,
\end{split}
\label{eq:r_u_derive}
\end{equation}
where $S_T^{-\cu}$ denotes the total index of $\bX_{-\cu}$. The first equality on the RHS of \eqref{eq:r_u_derive} follows from \eqref{eq:e_square_equality}.
The third and fourth equalities follow from the law of total variance, and the final inequality follows from the definition of the total index. It is evident that $R_{\cu} \in [0, 1]$, and $R_{\cu}$ quantifies the portion of the interaction effects involving the inputs in $\cu$ not captured by $\sum_{j \in -\cu}S^{\cu,j}$.

By utilizing $R_{\cu}$, we can further analyze the bias of $\E{V_{\tiny \mbox{L--NS}}^{\cu}}$ as follows:
	\begin{align*}
		\E{V_{\tiny \mbox{L--NS}}^{\cu}} - V &= \Var{\lhs{L}_N(\lhs{\bX}_\cu)} + \cO(K^{-1}) - V \\
		&= \Var{\E{\lhs{L}_N(\lhs{\bX}_\cu) \bigm| \lhs{\bX}_\cu}} + \E{\Var{\lhs{L}_N(\lhs{\bX}_\cu) \bigm| \lhs{\bX}_\cu}} + \cO(K^{-1}) - V \\
		&= V + \E{N^{-1}\Var{\lhs{\cY} \bigm| \lhs{\bX}_\cu} } \cdot R_{\cu} + r_N + \cO(N^{-2}) + \cO(K^{-1}) - V \\
		&= N^{-1}\E{\Var{\lhs{\cY} \bigm| \lhs{\bX}_\cu}} \cdot R_{\cu} + r_N + \cO(N^{-2}) + \cO(K^{-1}) \ ,
	\end{align*}
	where the first equality follows from Lemma \ref{lem:V_bias}, and the third equality follows from Equation \eqref{eq:var_LN_stein}. Furthermore, if $|\cu|=1$, it follows from Lemma \ref{lem:V_bias} that
		$\E{V_{\tiny \mbox{L--NS}}^{\cu}} - V = N^{-1}\E{\Var{\lhs{\cY} \mid \lhs{\bX}_\cu}} \cdot R_{\cu} + r_N + \cO(N^{-2}) + \cO(K^{-3}) + \cO(T^{-1})$.
\end{proof}

\subsubsection{Proof of Proposition \ref{prop:lhs_nest_var}}\label{app:proof_nest_lhs_var}
\begin{proof}
	To analyze the variance of $V_{\tiny \mbox{L--NS}}^{\cu}$, we use the following decomposition:
	\begin{equation}
\begin{split}
\Var{V_{\tiny \mbox{L--NS}}^{\cu}}
=&\ \Var{\frac{1}{K}\sum_{i=1}^{K}
\left( \lhs{L}_N(\lhs{\bX}_{\cu,i}) -
\frac{1}{K}\sum_{i=1}^{K} \lhs{L}_N(\lhs{\bX}_{\cu,i}) \right)^2 }  \\
=&\ \Var{\E{ \frac{1}{K} \sum_{i=1}^{K}
\left( \lhs{L}_N(\lhs{\bX}_{\cu,i}) -
\frac{1}{K}\sum_{i=1}^{K} \lhs{L}_N(\lhs{\bX}_{\cu,i}) \right)^2
\Bigg|\ \{\lhs{\bX}_{\cu,i}\}_{i \in[K]} }}  \\
&+ \E{\Var{ \frac{1}{K} \sum_{i=1}^{K}
\left( \lhs{L}_N(\lhs{\bX}_{\cu,i}) -
\frac{1}{K}\sum_{i=1}^{K} \lhs{L}_N(\lhs{\bX}_{\cu,i}) \right)^2
\Bigg|\ \{\lhs{\bX}_{\cu,i}\}_{i \in[K]} }}  \\
\coloneqq &\ \text{item } (i) + \text{item } (ii).
\end{split}
\label{eq:variance_lhs_decompose}
\end{equation}

	We first analyze item $(i)$ in \eqref{eq:variance_lhs_decompose}. It follows that
	\begin{align}
		&\Var{\E{ \frac{1}{K} \sum_{i=1}^{K}\left( \lhs{L}_N(\lhs{\bX}_{\cu,i}) - \frac{1}{K}\sum_{i=1}^{K} \lhs{L}_N(\lhs{\bX}_{\cu,i}) \right)^2 \Bigg|\ \{\lhs{\bX}_{\cu,i}\}_{i \in[K]} }} \nonumber \\
  		= & \Var{\frac{1}{K}\sum_{i=1}^{K}\E{ \lhs{L}^2_N(\lhs{\bX}_{\cu,i}) \mid \lhs{\bX}_{\cu,i} } -  \E{\left( \frac{1}{K}\sum_{i=1}^{K} \lhs{L}_N(\lhs{\bX}_{\cu,i}) \right)^2 \Bigg|\ \{\lhs{\bX}_{\cu,i}\}_{i \in[K]} } } \nonumber \\
  		\leq & 2  \Var{\frac{1}{K} \sum_{i=1}^{K} \E{  \lhs{L}^2_N(\lhs{\bX}_{\cu,i})  \Big|\ \{\lhs{\bX}_{\cu,i}\}_{i \in[K]}} }   + 2 \Var{\E{\left( \frac{1}{K}\sum_{i=1}^{K} \lhs{L}_N(\lhs{\bX}_{\cu,i}) \right)^2 \Bigg|\ \{\lhs{\bX}_{\cu,i}\}_{i \in[K]} }}  \nonumber \\
  		\coloneqq & 2 \cdot \mbox{item } (iii) +  2 \cdot \mbox{item } (iv) \nonumber \ .
 	\end{align}
 	Item $(iii)$ above is of order $\cO(K^{-1})$. Furthermore, when $|\cU|=1$, it reduces to $\cO(K^{-3})$ under LHS. Regarding item $(iv)$ above, we have
 	\begin{align}
 		\label{eq:var_exp_LN}
 		& \Var{\E{\left( \frac{1}{K}\sum_{i=1}^{K} \lhs{L}_N(\lhs{\bX}_{\cu,i}) \right)^2 \Bigg|\ \{\lhs{\bX}_{\cu,i}\}_{i \in[K]} }}  \nonumber \\
 		= & \Var{\Var{\frac{1}{K}\sum_{i=1}^{K} \lhs{L}_N(\lhs{\bX}_{\cu,i}) \ \Bigg|\ \{\lhs{\bX}_{\cu,i}\}_{i \in[K]} }  + \Ep{\frac{1}{K}\sum_{i=1}^{K} \lhs{L}_N(\lhs{\bX}_{\cu,i}) \ \Bigg|\ \{\lhs{\bX}_{\cu,i}\}_{i \in[K]}}{2} } \nonumber \\
 		= & \Var{\frac{1}{K^2} \sum_{i=1}^{K}\Var{\lhs{L}_N(\lhs{\bX}_{\cu,i}) \ \Big|\ \lhs{\bX}_{\cu,i} } + \left( \frac{1}{K}\sum_{i=1}^{K} L(\lhs{\bX}_{\cu,i})  \right)^2 }   \\
 		\leq &  \frac{2}{K^2}\Var{\frac{1}{K} \sum_{i=1}^{K}\Var{\lhs{L}_N(\lhs{\bX}_{\cu,i}) \ \Big|\ \lhs{\bX}_{\cu,i} }}  + 2  \Var{\left( \frac{1}{K}\sum_{i=1}^{K} L(\lhs{\bX}_{\cu,i})  \right)^2} \nonumber \\
 		= & \cO(K^{-3})  + 2 \Var{\left( \frac{1}{K}\sum_{i=1}^{K} L(\lhs{\bX}_{\cu,i})  \right)^2} \nonumber \\
 		= & \cO(K^{-3}) + 2 \cdot \mbox{item}(v) \nonumber \ ,
 	\end{align}
 	where the second equality on the RHS of \eqref{eq:var_exp_LN} follows from Lemma \ref{lem:L_unbias}.
 	Now we analyze item $(v)$ in \eqref{eq:var_exp_LN}. Let $\cL \coloneqq K^{-1}\sum_{i=1}^{K} L(\lhs{\bX}_{\cu,i})$, and recall that $\mu \coloneqq \E{\cY}$. Notice that $\E{\cL} = \mu$. It follows that
 	\begin{align*}
 		\Var{\cL^2} &= \E{\cL^4} - \Ep{\cL^2}{2}  \\
 		&= \E{(\cL - \mu)^4} + 4\mu\E{\cL-\mu}^3 + 6 \mu^2 \Var{\cL} + \mu^4 - (\mu^2 + \Var{\cL})^2 \\
 		&= (\kappa_{\cL} - 1) \Varp{\cL}{2} + 4\mu\E{\cL-\mu}^3 + 4 \mu^2 \Var{\cL} \ ,
 	\end{align*}
 	where $\kappa_{\cL}$ denotes the kurtosis of $\cL$. By applying the Cauchy--Schwarz inequality, we obtain the bound
 	\[\Var{\cL^2} \leq (\kappa_{\cL} - 1) \Varp{\cL}{2} + 4|\mu|\sqrt{\kappa_{\cL}}\Varp{\cL}{\frac{3}{2}}  + 4 \mu^2 \Var{\cL} \ .\]
 	Since $\Var{\cL} = \cO(K^{-1})$, item $(v)$ in \eqref{eq:var_exp_LN} is of order $\cO(K^{-1})$. Furthermore, when $|\cu| = 1$, we have $\Var{\cL} = \cO(K^{-3})$, and hence item $(v)$ is of order $\cO(K^{-3})$.

 	We now analyze item $(ii)$ in \eqref{eq:variance_lhs_decompose}. It follows that
 	\begin{align}
 		&\Esquare{\Var{ \frac{1}{K} \sum_{i=1}^{K}\left( \lhs{L}_N(\lhs{\bX}_{\cu,i}) - \frac{1}{K}\sum_{i=1}^{K} \lhs{L}_N(\lhs{\bX}_{\cu,i}) \right)^2 \ \middle|\ \{\lhs{\bX}_{\cu,i}\}_{i \in[K]} }} \nonumber \\
 		=& \Esquare{ \frac{1}{K-1} \E{ \left(  \lhs{L}_N(\lhs{\bX}_{\cu}) - L(\lhs{\bX}_{\cu}) \right)^4 \ \middle|\ \lhs{\bX}_{\cu} } - \frac{K-3}{(K-1)^2} \Varp{\lhs{L}_N(\lhs{\bX}_{\cu}) \ \middle|\ \lhs{\bX}_{\cu}}{2}} \nonumber \\
 		=& \frac{1}{K-1} \E{ \Esquare{ \left(  \lhs{L}_N(\lhs{\bX}_{\cu}) - L(\lhs{\bX}_{\cu}) \right)^4 \ \middle|\ \lhs{\bX}_{\cu} }} - \frac{K-3}{(K-1)^2} \Esquare{\Varp{\lhs{L}_N(\lhs{\bX}_{\cu}) \ \middle|\ \lhs{\bX}_{\cu}}{2}} \nonumber \\
 		=& \frac{1}{K-1} \Esquare{ \Var{ \left(  \lhs{L}_N(\lhs{\bX}_{\cu}) - L(\lhs{\bX}_{\cu}) \right)^2 \ \middle|\ \lhs{\bX}_{\cu} } }  + \frac{2}{(K-1)^2} \Esquare{\Varp{\lhs{L}_N(\lhs{\bX}_{\cu}) \ \middle|\ \lhs{\bX}_{\cu}}{2}} \nonumber \\
 		 \coloneqq & \frac{1}{K-1} \cdot \mbox{ item } (vi) +  \frac{2}{(K-1)^2} \Esquare{\Varp{\lhs{L}_N(\lhs{\bX}_{\cu}) \ \middle|\ \lhs{\bX}_{\cu}}{2}} \ ,
 		 \label{eq:variance_lhs_decompose_expect_var}
 	\end{align}
 	where the second to last equality follows from the following identity:
 	\begin{align*}
 		& \Esquare{ \E{ \left(  \lhs{L}_N(\lhs{\bX}_{\cu}) - L(\lhs{\bX}_{\cu}) \right)^4 \ \middle|\ \lhs{\bX}_{\cu} } } \\
 		=& \Esquare{ \Var{ \left(  \lhs{L}_N(\lhs{\bX}_{\cu}) - L(\lhs{\bX}_{\cu}) \right)^2 \ \middle|\ \lhs{\bX}_{\cu} } } + \Esquare{\Varp{\lhs{L}_N(\lhs{\bX}_{\cu}) \ \middle|\ \lhs{\bX}_{\cu}}{2} } \ .
 	\end{align*}

 	The second term on the RHS of \eqref{eq:variance_lhs_decompose_expect_var} is of order $\cO(T^{-2})$, since $\Esquare{\Varp{\lhs{L}_N(\lhs{\bX}_{\cu}) \ \middle|\ \lhs{\bX}_{\cu}}{2} } = \cO(N^{-2})$. Regarding item $(vi)$ in \eqref{eq:variance_lhs_decompose_expect_var}, we have
 	\begin{equation}\label{eq:conditional_var_bound}
 		\Esquare{ \Var{ \left(  \lhs{L}_N(\lhs{\bX}_{\cu}) - L(\lhs{\bX}_{\cu}) \right)^2 \bigm| \lhs{\bX}_{\cu} } }
 		\leq 2\Esquare{\Var{  \lhs{L}^2_N(\lhs{\bX}_{\cu})  \mid \lhs{\bX}_{\cu} }} + 2\Esquare{4L^2(\lhs{\bX}_{\cu})\Var{  \lhs{L}_N(\lhs{\bX}_{\cu})  \mid \lhs{\bX}_{\cu} } } \ .
 	\end{equation}

 	The second term on the RHS of \eqref{eq:conditional_var_bound} is of order $\cO(N^{-1})$. Similarly, the first term on the RHS of \eqref{eq:conditional_var_bound} is also $\cO(N^{-1})$, following the same reasoning as in the analysis of item $(v)$ in \eqref{eq:var_exp_LN}. Therefore, item $(ii)$ in \eqref{eq:variance_lhs_decompose} is of order $\cO(T^{-1})$. Combining the orders of items $(i)$ and $(ii)$ in \eqref{eq:variance_lhs_decompose} yields that the variance of the L-NS estimator is of order $\cO(K^{-1})$, and specifically $\cO(K^{-1}) + \cO(T^{-1})$ when $|\cU| = 1$.
\end{proof}

\subsection{Proofs in Subsection \ref{subsec:bias_correct_nest}}
\subsubsection{Proof of Proposition \ref{prop:bias_sj_lhs}}\label{app:proof_sj_lhs_bias}

\begin{proof}
To analyze the bias of $V_{\tiny \mbox{L--SJ}}^{\cu}$, we first note that
\begin{align*}
	\E{ \frac{1}{K} \sum_{i=1}^{K}\left(\lhs{L}_N(\bX_{\cu,i}) - \widehat{\mu}\right)^2} &= \Var{\lhs{L}_N(\bX_\cu)}  + \Var{\widehat{\mu}} = \Var{\lhs{L}_N(\bX_\cu)} + \cO(J^{-1}) \ ,
\end{align*}
	 where we recall that $\widehat{\mu}$ is the sample mean of $\cD_{\mbox{\tiny pre}}$ with size $J$. Hence, the bias of the L--SJ estimator follows as
\begin{align*}
	\E{V_{\tiny \mbox{L--SJ}}^{\cu}}- V
	=& \ I \cdot \Var{\lhs{L}_N(\bX_\cu)} - \frac{(I-1)}{I} \sum_{l=1}^{I}\Var{\lhs{L}_{N, -l}(\bX_\cu)}    + \Var{\widehat{\mu}} - V \\
	=& \ I \cdot \left(V + \frac{\E{\Var{\cY \mid \bX_\cu}}}{N} \cdot R_\cu + r_N + \cO(N^{-2})\right) \\
	& \ - \frac{(I-1)}{I} \sum_{l=1}^{I} \left(V+ \frac{\E{\Var{\cY \mid \bX_\cu}}}{N-N/I} \cdot R_\cu + r_N + \cO(N^{-2})  \right) + \cO(J^{-1}) - V  \\
	=& \  \cO(J^{-1}) + r_N + \cO(N^{-2}) \ ,
\end{align*}
where the second equality follows from Proposition \ref{prop:lhs_nest_bias} and Lemma \ref{lem:additive_approx}.
\end{proof}

\subsubsection{Proof of Proposition \ref{prop:bias_oh_lhs}}\label{app:proof_oh_lhs_bias}
\begin{proof}
	 The bias of the L--OH estimator can be expressed as
\begin{align*}
	\E{V_{\tiny \mbox{L--OH}}^{\cu}}- V &= \E{V_{\tiny \mbox{L--NS}}^{\cu}} - \frac{1}{N(N-1)K}\sum_{i=1}^{K}\sum_{j=1}^{N}\E{(\lhs{\cY}_{ij}-\lhs{L}_{N}(\lhs{\bX}_{\cu,i}))^2} - V\\
	&= \Var{\lhs{L}_N(\bX_\cu)} - \frac{1}{NK}\sum_{i=1}^{K}\E{\frac{1}{N-1}\sum_{j=1}^{N}\big(\lhs{\cY}_{ij}-\lhs{L}_{N}(\lhs{\bX}_{\cu,i})\big)^2} + \cO(K^{-1})- V \\
	&= V + \frac{\E{\Var{\cY \mid \bX_\cu}}}{N} \cdot R_\cu -\frac{\E{\Var{\cY \mid \bX_\cu}}}{N} \cdot R_\cu  + r_N + \cO(N^{-2}) + \cO(K^{-1}) - V \\
	&= \cO(K^{-1})+ r_N + \cO(N^{-2}) \ ,
\end{align*}
where the third equality follows from Proposition \ref{prop:lhs_nest_bias}.
\end{proof}

\subsubsection{Proof of Proposition \ref{prop:bias_jack_lhs}}\label{app:proof_jk_lhs_bias}
\begin{proof}
	 The bias of the L--JK estimator can be expressed as follows:
\begin{align*}
	\E{V_{\tiny \mbox{L--JK}}^{\cu}}- V =& \ I \E{V_{\tiny \mbox{L--NS}}^{\cu}} - (I-1)\E{\frac{1}{I}\sum_{l=1}^{I} V_{\tiny \mbox{L--NS},-l}^{\cu}} - V  \\
	=& \ I \cdot \Var{\lhs{L}_N(\bX_\cu)} - \frac{(I-1)}{I} \sum_{l=1}^{I}\Var{\lhs{L}_{N, -l}(\lhs{\bX}_\cu)} + \cO(K^{-1}) - V \\
	=& \ I \cdot \left(V + \frac{\E{\Var{\cY \mid \lhs{\bX}_\cu}}}{N} \cdot R_\cu + r_N + \cO(N^{-2})\right) \\
	& \ - \frac{(I-1)}{I} \sum_{l=1}^{I} \left(V+ \frac{\E{\Var{\cY \mid \lhs{\bX}_\cu}}}{N-N/I} \cdot R_\cu +r_N + \cO(N^{-2})  \right) + \cO(K^{-1}) - V \\
	=& \  \cO(K^{-1}) +  r_N + \cO(N^{-2}) \ ,
\end{align*}
where the third equality follows from Proposition \ref{prop:lhs_nest_bias} and Lemma \ref{lem:additive_approx}.
\end{proof}

\subsubsection{Proof of Proposition \ref{prop:var_jack_lhs}}\label{app:proof_jk_lhs_var}
\begin{proof}

To analyze the variance of $V_{\tiny \mbox{L--JK}}^{\cu}$, define $b^{(l)} \coloneqq V_{\mbox{\tiny L--NS}}^{\cu} - V_{\mbox{\tiny L--NS},-l}^{\cu}$ for each $l \in [I]$. We have
\begin{align}
	\Var{V_{\mbox{\tiny L--JK}}^{\cu}} =& \ \Var{I V_{\mbox{\tiny L--NS}}^{\cu} - \frac{I-1}{I}\sum_{l=1}^{I}V_{\mbox{\tiny L--NS},-l}^{\cu}} = \Var{V_{\mbox{\tiny L--NS}}^{\cu} + \frac{I-1}{I}\sum_{l=1}^{I}b^{(l)}} \nonumber \\
	= & \ \Var{V_{\mbox{\tiny L--NS}}^{\cu}} + \frac{(I-1)^2}{I}\Var{b^{(1)}} \label{eq:ljkvar_decompose_var} \\
	 & \ + 2(I-1) \Cov{V_{\mbox{\tiny L--NS}}^{\cu}}{b^{(1)}} + \frac{(I-1)^3}{I}\Cov{b^{(1)}}{b^{(2)}}. \label{eq:ljkvar_decompose_cov}
\end{align}
From the proof in Appendix \ref{app:proof_nest_lhs_var}, it follows that the first term in \eqref{eq:ljkvar_decompose_var} satisfies $\Var{V_{\mbox{\tiny L--NS}}^{\cu}} = \cO(K^{-1})$.

The term $\Var{b^{(1)}}$ in \eqref{eq:ljkvar_decompose_var} can be rewritten as
\begin{equation}\label{eq:variance_lbl_rewritten}
	\Var{b^{(1)}} = \Var{V_{\mbox{\tiny L--NS}}^{\cu} - V_{\mbox{\tiny L--NS},-1}^{\cu}} = \E{\left(V_{\mbox{\tiny L--NS}}^{\cu} - V_{\mbox{\tiny L--NS},-1}^{\cu}\right)^2} - \mathbb{E}^2\left( V_{\mbox{\tiny L--NS}}^{\cu} - V_{\mbox{\tiny L--NS},-1}^{\cu} \right).
\end{equation}

Since
\begin{equation}\label{eq:ljk_ns_diff_rate}
	\E{V_{\mbox{\tiny L--NS}}^{\cu} - V_{\mbox{\tiny L--NS},-1}^{\cu}} = \Var{\lhs{L}_N(\bX_\cu)} - \Var{\lhs{L}_{N,-1}(\lhs{\bX}_\cu)} = - \frac{1}{N(I-1)} \E{\Var{\cY \bigm| \lhs{\bX}_\cu}},
\end{equation}
it follows that $\Ep{V_{\mbox{\tiny L--NS}}^{\cu} - V_{\mbox{\tiny L--NS},-1}^{\cu}}{2} = N^{-2} \cdot(I-1)^{-2}  \cdot \Ep{\Var{\cY \bigm| \lhs{\bX}_\cu}}{2}$.

Regarding $\E{\left(V_{\mbox{\tiny L--NS}}^{\cu} - V_{\mbox{\tiny L--NS},-1}^{\cu} \right)^2}$, we have
\begin{align}
	\E{\left(V_{\mbox{\tiny L--NS}}^{\cu} - V_{\mbox{\tiny L--NS},-1}^{\cu} \right)^2} \leq \ & 2 \E{V_{\mbox{\tiny L--NS}}^{\cu} - \Var{\lhs{L}_N(\lhs{\bX}_\cu)} }^2  + 2 \E{V_{\mbox{\tiny L--NS},-1}^{\cu} - \Var{\lhs{L}_{N,-1}(\lhs{\bX}_\cu)}}^2 \nonumber \\
	&+ 2\left[\Var{\lhs{L}_N(\lhs{\bX}_\cu)} - \Var{\lhs{L}_{N,-1}(\lhs{\bX}_\cu)} \right]^2. \label{eq:lesquare}
\end{align}
We next analyze the three terms on the RHS of \eqref{eq:lesquare}. First, $\E{ V_{\mbox{\tiny L--NS}}^{\cu} - \Var{\lhs{L}_N(\lhs{\bX}_\cu)} }^2 = \Var{V_{\mbox{\tiny L--NS}}^{\cu}} = \cO(K^{-1})$. Similarly, $\E{ V_{\mbox{\tiny L--NS},-1}^{\cu} - \Var{\lhs{L}_{N,-1}(\lhs{\bX}_\cu)} }^2 =  \cO(K^{-1})$. Lastly, $\bigg(\Var{\lhs{L}_N(\lhs{\bX}_\cu)} -\Var{\lhs{L}_{N,-1}(\lhs{\bX}_\cu)} \bigg)^2 = \cO(N^{-2}) $. It follows from \eqref{eq:lesquare} that
\begin{equation}\label{eq:ljack_ns_different_second_moment_rate}
	\E{\left(V_{\mbox{\tiny L--NS}}^{\cu} - V_{\mbox{\tiny L--NS},-1}^{\cu} \right)^2} = \cO(K^{-1}) + \cO(N^{-2}).
\end{equation}
Combining \eqref{eq:variance_lbl_rewritten}, \eqref{eq:ljk_ns_diff_rate}, and \eqref{eq:ljack_ns_different_second_moment_rate} yields that
\begin{equation}\label{eq:ljack_proof_e2}
	\Var{b^{(l)}} = \frac{a^{\prime}}{K} + \frac{b^{\prime}}{N^2} + o(K^{-1}) + o(N^{-2}), \quad \forall l \in [I]
\end{equation}
for some positive constants $a^{\prime}$ and $b^{\prime}$.

For $\Cov{V_{\mbox{\tiny L--NS}}^{\cu}}{b^{(1)}}$ in \eqref{eq:ljkvar_decompose_cov}, it follows from the Cauchy–Schwarz inequality that
\begin{equation}\label{eq:ljack_proof_e3}
	\Cov{V_{\mbox{\tiny L--NS}}^{\cu}}{b^{(1)}} \leq \sqrt{\Var{V_{\mbox{\tiny L--NS}}^{\cu}}\Var{b^{(1)}}} = \cO(\max \{K^{-1}, K^{-1/2}N^{-1}\}).
\end{equation}
Similarly,
\begin{equation}\label{eq:ljack_proof_e4}
	\Cov{b^{(1)}}{b^{(2)}} \leq \sqrt{\Var{b^{(1)}}\Var{b^{(2)}}} = \cO(\max \{K^{-1}, N^{-2}\}).
\end{equation}
Combining \eqref{eq:ljack_proof_e2} through \eqref{eq:ljack_proof_e4} yields
	$\Var{V_{\mbox{\tiny L--JK}}^{\cu}} = a K^{-1}+ b N^{-2} + o(K^{-1}) +o(N^{-2})$
for some positive constants $a$ and $b$.
\end{proof}

\section{Additional Details for Section \ref{sec:experiment}}
\label{app:add_info_numerical}

\paragraph{\normalsize\emph{Ishigami function}} This example admits closed-form expressions for the first-order Sobol' indices \citep{gamboa2016statistical,ishigami1990importance}. The true values are $S^{1}=0.3139$, $S^{2}=0.4424$, and $S^{3}=0$.

\paragraph{\normalsize\emph{$g$-function}} This example also has closed-form first-order Sobol' indices \citep{tarantola2007estimating}. In the 3D case, the true values are $S^{1}=0.0476$, $S^{2}=0.1904$, and $S^{3}=0.7616$. In the 5D case, the true values are $S^{1}=0.48257$, $S^{2}=0.21443$, $S^{3}=0.12091$, $S^{4}=0.077382$, and $S^{5}=0.053866$.

\paragraph{\normalsize\emph{Hydrological model (hymod)}} Table \ref{app:hymod_input} summarizes the input variables of the hymod example, including their physical interpretations, probability distributions, and units.
Since the Sobol' indices for this example do not have an analytical form, we estimate them using the PF estimator in \eqref{eq:pick_freeze} with a total computational budget $T=10^8$ and treat the resulting estimates as ground truth. The estimated true values are $S^{1}=0.0411$, $S^{2}=0.0118$, $S^{3}=0.2657$, $S^{4}=0.0389$, and $S^{5}=0.2252$.

 \begin{table}[h!]
 	\centering
 	 \footnotesize
 	\captionsetup{font=footnotesize, labelfont=bf}
 	\caption{Descriptions and distributions of input variables for the hymod example. The notation
 		$U(a,b)$ denotes a uniform distribution on $[a,b]$, and ``--'' denotes a quantity with no units.}
 	\begin{tabular}{l l c c}
 		\hline
 		Input & Description & Units & Distribution \\
 		\hline
 		{\tt Sm} $(X_1)$   & maximum soil moisture                 & mm        & $U(0,400)$ \\
 		{\tt beta} $(X_2)$  & exponent in the soil moisture routine & --        & $U(0,2)$   \\
 		{\tt alfa}  $(X_3)$ & partition coefficient                  & --        & $U(0,1)$   \\
 		{\tt Rs}   $(X_4)$  & slow reservoir coefficient             & day$^{-1}$ & $U(0,0.1)$ \\
 		{\tt Rf}  $(X_5)$   & fast reservoir coefficient             & day$^{-1}$ & $U(0.1,1)$ \\
 		\hline
 	\end{tabular}
 	\label{app:hymod_input}
 \end{table}

}

\end{document}